%% file: paper.tex
\documentclass[a4paper,runningheads]{llncs}

\IfFileExists{cropped_pages}{
\paperwidth 13.2cm
\paperheight 21.3cm
\oddsidemargin  -0.85in
\evensidemargin -0.85in
\topmargin      -0.85in
}{}

\input{preamble}

\usepackage[colorlinks=true]{hyperref}
\hypersetup{%
  pdftitle={Tight Cutoffs for Guarded Protocols with Fairness},%
  pdfauthor={Simon Ausserlechner, Swen Jacobs, Ayrat Khalimov},%
  pdfsubject={},%
  pdfkeywords={\listOfKeywords},%
  colorlinks%
}

\title{Tight Cutoffs for\\ Guarded Protocols with Fairness}
\author{Simon Au{\ss}erlechner\inst{1}, Swen Jacobs\inst{2}, Ayrat Khalimov\inst{1}}
\institute{$^1$ IAIK, Graz University of Technology, Austria\\
					 $^2$ Reactive Systems Group, Saarland University, Germany}

\authorrunning{S. Au{\ss}erlechner, S. Jacobs and A. Khalimov}

\begin{document}
\iffinal
\else
  \textbf{TODOs:}
  \begin{itemize}
	\item \textbf{check conjunctive proofs}
	\item add (short) correctness arguments to proofs
	\item {\bf presentation of proof ideas in main part:}
	  \begin{itemize}
      \item pictures of main constructions?
	  \item explain why proofs do not work in general (for conj)
	  \end{itemize}
	\item \textbf{Notation}: 
	\begin{itemize}
	\item quantification over processes: unify $\forall i$ vs. $\forall p$, or explain difference
    \- we use 'conjunctive guarded system' and 'conjunctive system' -- unify SJ: we only have it in related work now, and cannot define 'conjunctive system' before
	\end{itemize}
	\item synthesize the running example
	\end{itemize}
	\textbf{Optional/Extensions:}
	\begin{itemize}
	\item tighten disjunctive deadlock detection
	\item conjunctive systems: extend to disjoint guards
	\item label-based cutoffs
	\item do we run into undecidability somewhere?\\ (label-based local deadlock detection?) (maybe: PMCP of push-down systems communicating via conj guards -- because for comm primitive $=$ unnested locks this is undec -- see Kahlon's CAV05 (reasoning about threads...) )
    \item disj: disj guards can be expressed in Petri nets. We can reduce to petri nets if we want to solve Global Deadlocks problem (anything else?)
    \item we consider the problem ``no global deadlocks $\land$ no local deadlocks under strong fairness $\land$ properties under undcond fairness''.
          An alternative is: ``no global deadlocks $\land$ properties under strong fairness''.
	\end{itemize}

\newpage

  \setcounter{tocdepth}{3}
  \tableofcontents
\fi

%
%
%
%

\maketitle              
 \begin{abstract}
Guarded protocols were introduced in a seminal
paper by Emerson and Kahlon (2000), and describe systems of processes whose transitions
are enabled or disabled depending on the existence of other processes in
certain local states. We study parameterized model checking and synthesis of
guarded protocols, both aiming at formal correctness arguments for systems with any number of processes. 
Cutoff results reduce reasoning about systems with an arbitrary 
number of processes to systems of a determined, fixed size. Our work stems 
from the observation that existing cutoff results for guarded protocols i) are 
restricted to closed systems, and ii) are of limited use for liveness 
properties because reductions do not preserve fairness. We close these gaps
and obtain new cutoff results for open systems with liveness properties under
fairness assumptions. Furthermore, 
we obtain cutoffs for the detection of global and local deadlocks, which are of 
paramount importance in synthesis. Finally, we prove tightness or asymptotic tightness for the new cutoffs.
\end{abstract} 

\input{intro}
\input{related}
\input{prelim}
\input{paramsynt}

\input{cutoffs}
\iffinal
\else
 \input{experiments}
\fi
\input{conclusion}
\bibliographystyle{splncs03}
\bibliography{paper,local,references,crossrefs}
\appendix
\input{appendix}

\end{document}

%% file: preamble.tex
\newif\iffinal
  \finaltrue

\newif\ifwithextensions
\withextensionsfalse

\usepackage{amsthm}
\usepackage{amsfonts}
\usepackage{amsmath}
\usepackage{amssymb}
\usepackage{temporal}
\usepackage{cite}
\usepackage{graphicx}
\usepackage{subfig}
\usepackage{enumitem}
\usepackage{wrapfig}
\usepackage{environ}
\usepackage{multirow}
\usepackage{xspace}
\usepackage{color,soul}
  \definecolor{lightblue}{rgb}{.8,.95,1}

\usepackage{thmtools, thm-restate}

\usepackage{tikz}
\usepackage{pgf}
\usetikzlibrary{positioning,calc,automata,arrows,fit,shapes}

\usepackage{caption}
\usepackage{mdframed}

\let\note\relax

\iffinal
  \usepackage[finalold]{trackchanges} 
\else
  \usepackage[inline]{trackchanges}
\fi

\addeditor{$S\!J$}
\addeditor{$S\!A$}
\addeditor{$A\!K$}
\addeditor{$R\!B$}
\newcommand{\sj}[1]{\note[$S\!J$]{#1}}

\newcommand{\ak}[1]{\note[$A\!K$]{#1}}

\newcommand{\listOfKeywords}{parameterized model checking, reactive synthesis, distributed systems, guarded protocols, deadlocks, fairness}

\newcommand{\mailto}[1]{\href{mailto:#1}{\nolinkurl{#1}}}

\newcommand{\mC}{\mathcal{C}}
\newcommand{\mD}{\mathcal{D}}
\newcommand{\mF}{\mathcal{F}}

\newcommand{\mP}{\mathcal{P}}
\newcommand{\mB}{\mathcal{B}}

\newcommand{\bbN}{\mathbb{N}}
\newcommand{\bbB}{\mathbb{B}}

\newcommand{\init}{{\sf init}\xspace}
\newcommand{\templateI}{A}
\newcommand{\templateII}{B}

\newcommand{\inputs}{\Sigma}
\newcommand{\localin}{\sigma}
\newcommand{\globIn}{E}

\newcommand{\visited}{\mathsf{Visited}}
\newcommand{\visInf}[2]{\mathsf{Visited}^\inf_{#1}\!(#2)}
\newcommand{\visFin}[2]{\mathsf{Visited}^\fin_{#1}\!(#2)}

\newcommand{\witfirst}{\mathsf{wit\_first}}
\newcommand{\witsecond}{\mathsf{wit\_second}}
\newcommand{\witlast}{\mathsf{wit\_last}}
\newcommand{\first}{\mathsf{first}}

\newcommand{\destutter}{\ensuremath{\mathsf{destutter}}\xspace}
\newcommand{\bound}{b}

\newcommand{\sched}{{{\sf move}}}
\newcommand{\enabled}{{{\sf en}}}

\newcommand\LTLmX{\ensuremath{\mbox{\textsf{LTL}}\backslash\textsf{X}}}
\newcommand\LTL{\ensuremath{\mbox{\textsf{LTL}}}}

\definecolor{darkgreen}{rgb}{0,0.5,0}
\definecolor{darkblue}{rgb}{0,0,.5}
\definecolor{mygray}{gray}{.3}

\newcommand{\state}{q}
\newcommand{\initstate}{{\sf init}\xspace}
\newcommand{\stateset}{\expandafter\MakeUppercase\expandafter{\state}}
\newcommand{\State}{s}
\newcommand{\Stateset}{\expandafter\MakeUppercase\expandafter{\State}}

\newcommand{\trans}{\ensuremath{\delta}}
\newcommand{\Trans}{\ensuremath{\Delta}}
\newcommand{\labeling}{L}
\newcommand{\labelings}{\mathcal{\labeling}}

\renewcommand{\time}{m}
\newcommand{\dead}{{\sf dead}\xspace}

\newcommand{\myloop}[1]{\ensuremath{#1 \rightarrow \initstate \rightarrow #1}\xspace}

\newcommand{\card}[1]{\left| {#1} \right|}
\newcommand{\Nat}{\ensuremath{\mathbb{N}}}
\renewcommand{\|}{\mid}
\newcommand{\smi}{\!\setminus\!}
\newcommand{\cupdot}{\mathbin{\dot{\cup}}}

\newcommand{\smartpar}[1]{\medskip \noindent {\bf #1}}

\newcommand{\spec}{\Phi}

\newcommand{\impl}{\rightarrow}
\newcommand{\Impl}{\Rightarrow}
\newcommand{\Implied}{\Leftarrow}
\renewcommand{\iff}{\Leftrightarrow}

\declaretheorem[name=Theorem]{thm}

\newcommand{\disj}[1]{\ensuremath{\exists\{#1\}}}

\newcommand{\cutoffsys}{\ensuremath{(A,B)^{(1,c)}}\xspace}
\newcommand{\largesys}{\ensuremath{(A,B)^{(1,n)}}\xspace}
\newcommand{\interleave}{\ensuremath{\textit{interleave}}\xspace}
\newcommand{\last}{\mathsf{last}}

\newcommand{\qstar}{{\ensuremath{q^\star}}}

\newcommand{\fin}{\textit{fin}}
\newcommand{\pref}{\mathsf{pref}}

\renewcommand{\inf}{\textit{inf}}

\iffinal
\newcommand{\gray}[1]{}

\else
\newcommand{\gray}[1]{{\color{black!50} #1}}

\fi

\newcommand{\transition}[3]{#1 \stackrel{#3}{\rightarrow} #2}
\newcommand{\slice}[2]{\ensuremath{[{#1}\!:\!{#2}]}}

\renewcommand{\fair}{\textit{fair}}

\newcommand{\specialcellC}[2][c]{%
  \begin{tabular}[#1]{@{}c@{}}#2\end{tabular}}

\newcommand{\li}{\begin{itemize}}
\newcommand{\il}{\end{itemize}}
\renewcommand{\-}{\item}

\newcommand{\myparagraphraw}[1]{\smallskip\noindent{\bf#1}}
\newcommand{\myparagraph}[1]{\myparagraphraw{#1.}}

\NewEnviron{tikzLTS}{%
  \begin{tikzpicture}[node distance=2cm,inner sep=1pt,minimum size=0.5mm,bend angle=20]
	 	\tikzstyle{proc} = [rectangle,draw=black,fill=green!20,thick,inner sep=10pt]
	  	\tikzstyle{state} = [circle,draw=black,thick,inner sep=3pt]
	  	\tikzstyle{noproc} = [circle]
	  	\tikzstyle{lbl} = [rectangle,node distance=2cm]
  		\tikzstyle{pre} = [ <-,shorten <=2pt,shorten >=2pt, >=stealth', semithick]
	  	\tikzstyle{post} = [ ->,shorten <=2pt,shorten >=2pt, >=stealth', semithick]
    \BODY
  \end{tikzpicture}
}

\NewEnviron{tikzPath}{%
  \begin{tikzpicture}	
  	\tikzstyle{pstate} = [circle, fill=black,thick, inner sep=2pt, minimum size=2mm]
    \tikzstyle{hiddenstate} = [circle]
  	\tikzstyle{trans-notsched} = [ ->,shorten <=2pt,shorten >=2pt, >=stealth', semithick]
  	\tikzstyle{trans-sched} = [ ->,shorten <=2pt,shorten >=2pt, >=stealth', very thick] 
    \BODY
    \end{tikzpicture}
}

%% file: intro.tex
\section{Introduction}
\label{sec:intro}

Concurrent hardware and software systems are notoriously hard to get correct.
Formal methods like model 
checking or synthesis can be used to guarantee correctness, but the state explosion 
problem prevents us from using such methods for systems with a large number 
of components. Furthermore, correctness properties are often expected to hold 
for an \emph{arbitrary} number of components. Both problems can be solved by 
\emph{parameterized} model checking and synthesis approaches, which give correctness 
guarantees for systems with any number of components without considering every 
possible system instance explicitly.

While parameterized model checking (PMC) is undecidable in general~\cite{Suzuki88}, 
there exist a number of methods that decide the problem for specific classes 
of systems~\cite{German92,Emerson00,Emerso03}, as well as semi-decision 
procedures that are successful in many interesting 
cases~\cite{Kurshan95,Clarke08,KaiserKW10}.
In this paper, we consider the \emph{cutoff} method that can guarantee properties of 
systems of arbitrary size by considering only systems of up to a certain 
fixed size, thus providing a decision procedure for PMC if components are finite-state.

We consider systems that are composed of an arbitrary number of processes,
each an instance of a process template from a given, finite set.  
Process templates can be viewed as synchronization 
skeletons~\cite{EmersonC82}, i.e., program abstractions that suppress information not 
necessary for synchronization. In our system model, processes communicate by guarded updates, 
where guards are statements about other processes that are interpreted either 
conjunctively (``every other process satisfies the guard'') or disjunctively 
(``there exists a process that satisfies the guard''). Conjunctive guards can 
model atomic sections or locks, disjunctive guards can model token-passing or to some extent pairwise rendezvous (cf.~\cite{EmersonK03}). 

This class of systems has been studied by Emerson and 
Kahlon~\cite{Emerson00}, and cutoffs that depend on the 
size of process templates are known for specifications of the form 
$\forall{\bar{p}}.\ \spec(\bar{p})$, 
where $\spec(\bar{p})$ is an $\LTLmX$ property over the local states of one or more processes 
$\bar{p}$. Note that this does not allow us to specify fairness 
assumptions, for two reasons: i) to specify fairness, additional atomic propositions for enabledness and scheduling of processes are needed, and ii) specifications with global fairness assumptions are of the form 
$(\forall{\bar{p}}.\ \fair(\bar{p})) \impl (\forall{\bar{p}}.\ \spec(\bar{p}))$.
Because neither is supported by \cite{Emerson00}, the existing cutoffs are of 
limited use for 
reasoning about liveness properties. 

Emerson and Kahlon~\cite{Emerson00} mentioned this limitation and illustrated it
using the process template on the figure on the right.
Transitions from the initial state 
\begin{wrapfigure}{r}{0.32\linewidth}
\vspace{-20pt}
\hspace{-10pt}
\scalebox{0.8}{
\input{img/mutex}
}
\label{fig:mutex}
\vspace{-20pt}
\hspace{-10pt}
\end{wrapfigure}
$N$ to the ``trying'' state $T$, and from the critical state $C$ to $N$ are always 
possible, while the transition from $T$ to $C$ is only possible if no other 
process is in $C$. The existing cutoff results can be 
used to prove safety properties like mutual exclusion for systems composed of 
arbitrarily many copies of this template. However, they cannot be used to prove 
starvation-freedom properties like 
$\forall{p}. \pforall \always (T_p \rightarrow \eventually C_p)$,
stating that every process $p$ that enters its local state $T_p$ will 
eventually enter state $C_p$, because without fairness 
of scheduling the property does not hold. 

Also, Emerson and Kahlon~\cite{Emerson00} consider only closed systems. 
Therefore, in this example, processes always try to enter $C$. 
In contrast, in open systems the transition to $T$ might be a reaction 
to a corresponding input from the environment that makes entering $C$ necessary. While it is possible to convert an open system to a closed system that is equivalent under \LTL\ properties, this comes at the cost of a blow-up.

\smallskip
\noindent {\bf Motivation.} Our work is inspired by applications in
parameterized synthesis~\cite{Jacobs14}, where the goal is to automatically
construct process templates such that a given specification is satisfied in
systems with an arbitrary number of components. In this setting, one generally
considers \emph{open systems} that interact with an uncontrollable environment,
and most specifications contain liveness properties that cannot be guaranteed
without fairness assumptions. Also, one is in general interested in synthesizing
deadlock-free systems. \emph{Cutoffs} are essential for parameterized
synthesis, and we will show in Sect.~\ref{sec:paramsynt} how size-dependent
cutoffs can be integrated into the parameterized synthesis approach.

\smallskip
\noindent
{\bf Contributions.}
\begin{itemize}
\item We show that existing cutoffs for model checking of $\LTLmX$ properties are in general not sufficient for systems with \emph{fairness assumptions}, and provide new cutoffs for this case.
\item We improve some of the existing cutoff results, and give separate cutoffs for the problem of \emph{deadlock detection}, which is closely related to fairness. 
\item We prove \emph{tightness} or asymptotical tightness for all of our cutoffs, showing that smaller cutoffs cannot exist with respect to the parameters we consider.
\end{itemize}
Moreover, all of our cutoffs directly support \emph{open systems}, where each process may communicate with an adversarial environment. This makes the blow-up incurred by translation to an equivalent closed system unnecessary. 
The results presented here are based on a more detailed preliminary version of this paper~\cite{AJK15}.

%% file: img/mutex.tex
\begin{tikzpicture}[->,>=stealth',shorten >=1pt,auto,node distance=2cm,
                    semithick]
  \tikzstyle{every state}=[align=center, anchor=center, inner sep=2.2pt,minimum size=0.5mm]
  \tikzstyle{every edge} = [align=center,draw=black,]
  
  \node[state] (N) [label={below:$$},] {$N$};
  \node[state] (T) [right of=N, label={below:$$},] {$T$};
  \node[state] (C) [right of=T, label={below:$$},] {$C$};

  \path 
  (N) edge  node {$true$} (T)
  (T) edge node {$\forall \{T,N$\}} (C)
  (C) edge [bend left=24] node {$true$} (N);
\end{tikzpicture}

%% file: related.tex
\section{Related Work}
\label{sec:related}
\iffinal
\else
\hrule
\li
  \- cite ``Parameterized Model-Checking of Timed Systems with Conjunctive Guards'' and another Sasha's paper that contain results for disj guards (including undec result for some CTL*-like logic).

  \- Delzanno in ``Towards Automatic Verification of Java Programs'' studied Monitor objects, and had some kind of boolean guards -- related and compare with our conjunctive guards.

  \- relate with Kahlon's papers on locks, bounded lock chains, etc.

  \- relate to ``Using branching time temporal logic to synthesize synchronization skeletons''

  \- The paper `Automatic Deductive Verification with invisible invariants' \sj{mentioned it below}\ak{yes, but my point was different -- they do have cutoffs. (minor: not for the conference paper)} \sj{seems like they have cutoffs, but only for checking inductive invariants, not finding them} contains cutoff $2b+2$ for 1-indexed properties and $2b+3$ for 2-indexed (pp.7,8) . Is there a connection with our model/cutoffs?

  \- Benedikt Bollig has a paper on ParamSynt: comm is some form of comm automata.
     Can we reduce disj guards to comm automata?

\il
\hrule
\

\fi

As mentioned, we extend the results of Emerson and Kahlon~\cite{Emerson00} who study PMC of guarded protocols, but do not support fairness assumptions, nor provide cutoffs for 
deadlock detection.
In \cite{EmersonK03} they extended their work to systems with limited forms of guards and broadcasts, and also proved undecidability 
of PMC of conjunctive guarded protocols wrt. $\LTL$ (including $\nextt$), 
and undecidability wrt. $\LTLmX$ for systems with both 
conjunctive and disjunctive guards.

Bouajjani et al.~\cite{Bouajjani08} study parameterized model checking of
resource allocation systems (RASs). Such systems have a bounded number of resources, 
each owned by at most one process at any time. Processes are pushdown automata, 
and can request resources with high or normal priority.
RASs are similar to conjunctive guarded protocols in that certain
transitions are disabled unless a processes has a certain resource. 
RASs without priorities and with processes being finite state automata can 
be converted to conjunctive guarded protocols (at the price of blow up), 
but not vice versa. The authors study parameterized model checking wrt. 
$\LTLmX$ properties on arbitrary or on strong-fair runs, and (local or global) deadlock detection. The proof structure resembles that of \cite{Emerson00}, as does ours.
 
German and Sistla~\cite{German92} considered global deadlocks and strong fairness properties for systems with pairwise rendezvous communication in a clique.
Emerson and Kahlon~\protect\cite{EmersonK03} have shown that disjunctive guard systems can be reduced to such pairwise rendezvous systems.
However, German and Sistla \cite{German92} do not provide cutoffs, 
nor do they consider local deadlocks, 
and their specifications can talk about one process only.
Aminof et al.~\cite{AminofKRSV14} have 
recently extended these results to more general topologies, and have 
shown that for some decidable PMC problems there are 
no cutoffs, even in cliques. \ak{update the paragraphs, they also have some
  results on disj?}

\sj{we might remove this paragraph:}\ak{may be worth to compare with conj guards}
Emerson and Namjoshi provide cutoffs for systems that pass a valueless token 
in a ring~\cite{Emerso03}, 
which is essentially resource allocation of a 
single resource with a specific allocation scheme. Their results have been extended to more general topologies~\cite{Clarke04c,AminofJKR14}.
All of these results consider fairness of token passing in the sense that 
every process receives the token infinitely often.

Many of the decidability results above have recently been surveyed by Bloem
et al~\cite{BloemETAL15}. In addition, there are many methods based on
semi-algorithms. 

``Dynamic cutoff''
approaches~\cite{KaiserKW10,AbdullaHH13} support larger 
classes of systems, and try to find cutoffs for a 
concrete system and specification. These 
methods can find smaller cutoffs than those that are statically determined for 
a whole class of systems and specifications, but are currently limited to safety
properties.
The invisible invariants method~\cite{PnueliRZ01} tries to find
invariants in small systems, and applies a specialized cutoff result to prove
correctness of all instances, including an extension to liveness properties~\cite{FangPPZ06}.

Finally, there are methods that work completely without cutoffs, like regular model
checking~\cite{Bouajjani00}, network invariants~\cite{WolperL89,Kurshan95,KestenPSZ02},
and counter abstraction~\cite{PnueliXZ02}. They are in general
incomplete, but may provide decision procedures for certain classes of systems
and specifications, and support liveness to some extent.


%% file: prelim.tex
\section{Preliminaries}
\label{sec:prelim}\label{sec:definitions}

\sj{should we add a remark on which definitions are old, and which are new?}

\subsection{System Model}
\label{sec:model}

We consider systems $A {\parallel} B^n$, usually written $\largesys$, 
consisting of one copy of a process template $A$ and $n$ copies of a process template $B$, 
in an interleaving parallel composition.%
We distinguish objects that belong to different templates by indexing them with
the template. E.g., for process template $U \in \{A,B\}$, $Q_U$ is the set of
states of $U$. For this section, fix two disjoint finite sets $Q_A$, $Q_B$ as
sets of states of process templates $A$ and $B$, and a positive integer $n$.

\smartpar{Processes.} A \emph{process template} 
 is a transition system
  $U=(\stateset, \init, \inputs, \trans)$ with 
	\begin{itemize}
	\item $\stateset$ is a finite set of states including the
  initial state $\init$,
	\item $\inputs$ is a finite input alphabet,
	\item $\trans: \stateset \times \inputs \times \mP(Q_A \cupdot Q_B) \times \stateset$ is a guarded transition relation.
	\end{itemize}
A process template is \emph{closed} if $\inputs = \emptyset$, and otherwise \emph{open}.

We define the size $\card{U}$ of a process template $U \in \{A,B\}$ as $\card{\stateset_U}$. A copy of template $U$ will be called a \emph{$U$-process}.
Different $B$-processes are distinguished by subscript, i.e., for $i \in [1..n]$, $B_i$ is the $i$th copy of $B$, and $\state_{B_i}$ is a state of $B_i$. A state of the $A$-process is denoted by $q_A$. 

For the rest of this subsection, fix templates $A$ and $B$. We assume that $\inputs_A \cap \inputs_B = \emptyset$. We will also write $p$ for a process in $\{ A, B_1, \ldots, B_n\}$, unless $p$ is specified explicitly.

\smartpar{Disjunctive and Conjunctive Systems.}
In a system $\largesys$, consider global state $s = (\state_A,\state_{B_1},\ldots,\state_{B_n})$ and global input $e=(\localin_A,\localin_{B_1},\ldots,\localin_{B_n})$.
We also write $s(p)$ for $q_p$, and $e(p)$ for $\sigma_p$.
A local transition $(\state_p,\localin_p,g,\state_p') \in \trans_U$ of $p$ is \emph{enabled for $s$ and $e$} if its \emph{guard} $g$ is satisfied for $p$ in $s$, written $(s,p) \models g$. 
Disjunctive and conjunctive systems are distinguished by the \emph{interpretation of guards}:
\begin{align*}
\text{In disjunctive systems: } & (s,p) \models g \text{~~~iff~~~} 
\exists p' \in \{A,B_1,\ldots,B_n\} \setminus \{p\}:\ \ \state_{p'} \in g. \\
\text{In conjunctive systems: } & (s,p) \models g \text{~~~iff~~~} 
\forall p' \in \{A,B_1,\ldots,B_n\} \setminus \{p\}:\ \ \state_{p'} \in g.
\end{align*}

Note that we check containment in the guard (disjunctively or conjunctively) 
only for local states of processes \emph{different from} $p$. A process is \emph{enabled} for $s$ and $e$ if at least one of its transitions is enabled for $s$ and $e$, otherwise it is \emph{disabled}.

Like Emerson and Kahlon~\cite{Emerson00}, 
we assume that in conjunctive systems $\init_A$ and $\init_B$ are contained in all guards,
i.e., they act as neutral states.
Furthermore, we call a conjunctive system \emph{$1$-conjunctive} if every guard is of the form $(Q_A \cupdot Q_B) \setminus \{q\}$ for some $q \in Q_A\cupdot Q_B$.

Then, \largesys is defined as the transition 
system $(S,\init_S,\globIn,\Trans)$ with 
\begin{itemize}
\item set of global states $S = (\stateset_A) \times (\stateset_B)^{n}$, 
\item global initial state $\init_S = (\initstate_A,\initstate_B,\ldots,\initstate_B)$, 
\item set of global inputs $\globIn = (\inputs_A) \times (\inputs_B)^{n}$,
\item and global transition relation $\Trans \subseteq S \times \globIn \times S$ with $(s,e,s') \in \Trans$ iff 
\begin{enumerate}[label=\roman*)] 
  \item $s=(\state_A,\state_{B_1},\ldots,\state_{B_n})$, 
  \item $e=(\localin_A, \localin_{B_1},\ldots,\localin_{B_n})$, and 
  \item $s'$ is obtained from $s$ by replacing one local state $\state_p$ with a new local state $\state_p'$, where $p$ is a $U$-process with local transition $(\state_{p},\localin_{p},g,\state_p') \in \trans_U$ and $(s,p) \models g$. 
\end{enumerate}
\end{itemize}
We say that a system $\largesys$ is \emph{of type} $(A,B)$. It is called a
\emph{conjunctive system} if guards are interpreted conjunctively, and a
\emph{disjunctive system} if guards are interpreted disjunctively. \sj{is
  previous sentence necessary? conj. and disj. systems are already defined before}
A system is \emph{closed} if all of its templates are closed.
We often denote the set $\{B_1,...,B_n\}$ as $\mB$.

\smartpar{Runs.} 
A \emph{configuration} of a system is a triple $(s,e,p)$, where $s \in S$, $e 
\in \globIn$, and $p$ is either a system process, or the special symbol $\bot$.
 A \emph{path} of a system is a configuration sequence 
$x = (s_1,e_1,p_1),(s_2,e_2,p_2),\ldots$ such that for all $\time < |x|$ there is a 
transition $(s_\time,e_\time,s_{\time+1}) \in \Trans$ based on a local 
transition of process $p_\time$. We say that process 
$p_\time$ \emph{moves} at \emph{moment} $\time$. 
Configuration $(s,e,\bot)$ appears
 iff all processes are disabled for $s$ and $e$.
Also, for every $p$ and $\time < |x|$: 
either $e_{\time+1}(p) = e_\time(p)$ or process $p$ moves at moment $\time$. 
That is, the environment keeps input to each process unchanged until 
the process can read it.\footnote{By only considering inputs that are actually processed, we 
approximate an 
action-based semantics. Paths that do not fulfill this requirement are not 
very interesting, since the environment can violate any interesting 
specification that involves input signals by manipulating them when the 
corresponding process is not allowed to move.} 

A system \emph{run} is a maximal path starting in the initial state. Runs are either infinite, or they end in a configuration $(s,e,\bot)$. We say that a run is \emph{initializing} if every 
process
that moves infinitely often also visits 
its $\initstate$ 
infinitely often.

Given a system path $x = (s_1,e_1,p_1),(s_2,e_2,p_2),\ldots$ and a process $p$, the \emph{local path} of $p$ in $x$ is the projection $x(p) = (s_1(p),e_1(p)),(s_2(p),e_2(p)),\ldots$ of $x$ onto local states and inputs of $p$.
Similarly define the projection on two processes $p_1,p_2$ denoted by $x(p_1,p_2)$.


\smartpar{Deadlocks and Fairness.}
A run is \emph{globally deadlocked} if it is finite.
An infinite run is \emph{locally deadlocked} for process $p$ if there exists $\time$ such that $p$ is disabled for all $s_{\time'},e_{\time'}$ with $\time'\ge \time$. A run is \emph{deadlocked} if it is locally or globally deadlocked.
A system \emph{has a (local/global) deadlock} if it has a (locally/globally) deadlocked run. Note that absence of local deadlocks for all $p$ implies absence of global deadlocks, but not the other way around.

A run $(s_1,e_1,p_1), (s_2,e_2,p_2),...$ is \emph{unconditionally-fair} if every process moves infinitely often. 
A run is \emph{strong-fair} if it is infinite and for every process $p$, if $p$ is enabled infinitely often, then $p$ moves infinitely often.
We will discuss the role of deadlocks and fairness in synthesis in Sect.~\ref{sec:paramsynt}.

\begin{remark}
Why do we consider systems $A {\parallel} B^n$?
Emerson and Kahlon~\cite{Emerson00} showed how to generalize cutoffs 
for such systems to systems of the form $A^m {\parallel} B^n$, 
and further to systems with an arbitrary number of process templates 
$U_1^{n_1} {\parallel} \ldots {\parallel} U_m^{n_m}$.
This generalization also works for our new results, except for the cutoffs for deadlock detection that are restricted to 1-conjunctive systems (see Section~\ref{sec:cutoffs}).
\end{remark}

\subsection{Specifications}
\label{sec:semantics}
Fix templates $(A,B)$. We consider formulas in $\LTLmX$, i.e., $\LTL$ without the next-time operator $\nextt$.
Let $h(A,B_{i_1},\ldots,B_{i_k})$ be an $\LTLmX$ formula over atomic propositions from $Q_A \cup \Sigma_A$ and indexed propositions from $(Q_B \cup \Sigma_B) \times \{i_1,\ldots,i_k\}$. For a system $\largesys$ with $n \geq k$ and $i_j \in [1..n]$, satisfaction of $\pforall h(A,B_{i_1},\ldots,B_{i_k})$ and $\pexists h(A,B_{i_1},\ldots,B_{i_k})$ is defined in the usual way (see e.g. \cite{PrinciplesMC}).

\smartpar{Parameterized Specifications.} 	
\label{sec:parameterized}
A \emph{parameterized specification} is a temporal logic formula
with indexed atomic propositions and quantification over indices. 
We consider formulas of the forms
$\forall{i_1,\ldots,i_k.} \pforall h(A,B_{i_1},\ldots,B_{i_k})$ and\\ 
$\forall{i_1,\ldots,i_k.} \pexists h(A,B_{i_1},\ldots,B_{i_k})$. 
For given $n \geq k$, 
$$\largesys {\models} \forall{i_1,{\ldots},i_k.} \pforall h(A,B_{i_1},{\ldots},B_{i_k})$$
~iff~
$$\largesys {\models} \bigwedge_{j_1 \neq {\ldots} \neq j_k \in [1..n]} \pforall h(A,B_{j_1},{\ldots},B_{j_k}).$$ 
By symmetry of guarded protocols, this is equivalent 
(cp.\cite{Emerson00})
to $\largesys \models \pforall h(A,B_1,\ldots,B_k)$. 
The latter formula is denoted by $\pforall h(A,B^{(k)})$, 
and we often use it instead of the original $\forall{i_1,\ldots,i_k.} \pforall h(A,B_{i_1},...,B_{i_k})$. For formulas with path quantifier $\pexists$, satisfaction is defined analogously, and equivalent to satisfaction of $\pexists h(A,B^{(k)})$.

\smartpar{Specification of Fairness and Local Deadlocks.}
It is often convenient to express fairness assumptions and local deadlocks 
as parameterized specifications.
To this end,
define auxiliary atomic propositions $\sched_p$ and $\enabled_p$ for every process $p$ of system $(A,B)^{(1,n)}$. At moment $\time$ of a given run $(s_1,e_1,p_1),(s_2,e_2,p_2), \ldots$, let $\sched_p$ be true whenever $p_\time = p$, and let $\enabled_p$ be true if $p$ is enabled for $s_\time, e_\time$. Note that we only allow the use of these propositions to define fairness, but not in general specifications.
Then, an infinite run is 
\begin{itemize}
\item \emph{local-deadlock-free} if it satisfies $\forall{p}. \GF \enabled_p$, abbreviated as $\spec_{\neg dead}$,
\item \emph{strong-fair} if it satisfies $\forall{p}. \GF \enabled_p \impl \GF \sched_p$, abbreviated as $\spec_{strong}$, and 
\item \emph{unconditionally-fair} if it satisfies $\forall{p}. \GF \sched_p$, abbreviated as $\spec_{uncond}$.
\end{itemize}

If \emph{fair} is a fairness notion and 
$\pforall h(A,B^{(k)})$ 
a specification, then we write 
$\pforall_{fair} h(A,B^{(k)})$ for $\pforall (\spec_{fair} 
\rightarrow h(A,B^{(k)}))$. Similarly, we write $\pexists_{fair} h(A,B^{(k)})$ for $\pexists (\spec_{fair} \land h(A,B^{(k)}))$.

\subsection{Model Checking and Synthesis Problems}
\label{sec:nonparameterized_synthesis}
For a given system $\largesys$ and specification $h(A,B^{(k)})$ with $n \ge k$,
\begin{itemize}
\item the \emph{model checking problem} is to decide whether $\largesys \models \pforall h(A,B^{(k)})$,
\item the \emph{deadlock detection problem} is to decide whether $\largesys$
      does not have global nor local deadlocks,
\item the \emph{parameterized model checking problem} (PMCP) is to decide whether $\forall m \ge n:\ (A,B)^{(1,m)} \models \pforall h(A,B^{(k)})$, and 
\item the \emph{parameterized deadlock detection problem} is to decide whether 
      for all $m \ge n$, $(A,B)^{(1,m)}$ does not have global nor local deadlocks.
\end{itemize}
For a given number $n \in \bbN$ and specification $h(A,B^{(k)})$ with $n \ge k$,
\begin{itemize}
\item the \emph{template synthesis problem} is to find process templates $A,B$ such that
$\largesys \models \pforall h(A,B^{(k)})$ and $\largesys$ does not have global deadlocks. 
\item 
the \emph{bounded template synthesis problem} for a pair of bounds $(\bound_A,\bound_B) \in \bbN \times \bbN$ 
is to solve the template synthesis problem with 
$\card{A} \leq \bound_A$ and $\card{B} \leq \bound_B$.
\item the \emph{parameterized template synthesis problem} is to find process templates $A,B$ such that $\forall m \ge n:\ (A,B)^{(1,m)} \models \pforall h(A,B^{(k)})$ and $(A,B)^{(1,m)}$ does not have global deadlocks.
\end{itemize}
These definitions can be flavored with different notions of fairness 
(and similarly for the $\pexists$ path quantifier).
In the next section we clarify the problems studied.

%% file: paramsynt.tex
\section{Reduction Method and Challenges}
\label{sec:paramsynt}
 
We show how to use existing cutoff results of Emerson and Kahlon~\cite{Emerson00} to reduce the PMCP to a standard model checking problem, and parameterized synthesis to template synthesis. We note the limitations of the existing results that are crucial in the context of synthesis.

\smallskip
\noindent
{\bf Reduction by Cutoffs.}
A \emph{cutoff} for a system type $(\templateI,\templateII)$ and a specification $\spec$ is a number $c \in \bbN$ such that:
\[ 
\forall n \ge c: \left( \largesys \models \spec ~~\iff~~ \cutoffsys \models \spec \right).
\]
Similarly, $c \in \bbN$ is a \emph{cutoff for (local/global) deadlock detection}
if $\forall n \ge c:$ $\largesys$ has a (local/global) deadlock iff $\cutoffsys$ has a (local/global) deadlock.
For the systems and specifications presented in this paper, cutoffs can be computed from 
the size of process template $B$ and the number $k$ of copies of $B$ 
mentioned in the specification, and are given as expressions like 
$\card{B}+k+1$.

\begin{remark}\label{re:EK_cutoffs}
Our definition of a cutoff is different from that of Emerson and Kahlon~\cite{Emerson00}, and instead similar to, e.g., Emerson and Namjoshi~\cite{Emerso03}. The reason is that we want the following property to hold for any $(A,B)$ and $\Phi$: 

\noindent \emph{if $n_0$ is the smallest number such that ~$\forall n \geq n_0:\ \largesys \models \Phi$},

\noindent \emph{then any $c<n_0$ is not a cutoff, any $c\geq n_0$ is a cutoff.}

\smallskip\noindent 
We call $n_0$ the \emph{tight} cutoff.
The definition in \cite[page 2]{Emerson00} requires that $\forall{n\leq c}. \largesys \models \Phi$ if and only if $\forall{n \geq 1}: \largesys \models \Phi$, and thus allows stating $c<n_0$ as a cutoff if $\Phi$ does not hold for all $n$. \qed
\end{remark}
In model checking, a cutoff allows us to check whether any ``big'' system satisfies the specification by checking it in the cutoff system. 
As noted by Jacobs and Bloem~\cite{Jacobs14}, a 
similar reduction applies to the parameterized synthesis problem. 
For guarded protocols, we obtain the following 
\emph{semi-decision procedure for parameterized synthesis}\ak{is it decidable?}:
\begin{enumerate}
  \item[0.] set initial bound $(\bound_A,\bound_B)$ on size of process templates;
  \item[1.] determine cutoff for $(\bound_A,\bound_B)$ and $\spec$;
  \item[2.] solve bounded template synthesis problem for cutoff, size bound, and $\spec$;
  \item[3.] if successful return $(A,B)$ else increase $(\bound_A,\bound_B)$ and goto (1).
\end{enumerate}

\myparagraph{Existing Cutoff Results}
Emerson and Kahlon~\cite{Emerson00} have shown:

\begin{thm}[Disjunctive Cutoff Theorem] \label{thm:disj-cutoff-pairs}
    For closed disjunctive systems, $\card{B}+2$ is a cutoff {$^{(\dagger)}$} for formulas of the
    form $\pforall h(A,B^{(1)})$ and $\pexists h(A,B^{(1)})$, and for global
    deadlock detection.
\end{thm}
 
\begin{thm}[Conjunctive Cutoff Theorem] \label{thm:conj-cutoff}
    For closed conjunctive systems, $2\card{B}$ is a cutoff {$^{(\dagger)}$} for formulas of the
    form $\pforall h(A)$ and $\pexists h(A)$, and for global deadlock detection.
    For formulas of the form $\pforall h(B^{(1)})$ and $\pexists h(B^{(1)})$,
    $2\card{B}+1$ is a cutoff.
\end{thm}

\begin{remark}
$^{(\dagger)}$ Note that Emerson and Kahlon \cite{Emerson00} proved these results for
a different definition of a cutoff (see Remark \ref{re:EK_cutoffs}).  
Their results also hold for our definition, except possibly for
global deadlocks.  For the latter case to hold with the new cutoff definition, one 
also needs to prove the direction ``global deadlock in the cutoff system implies global
deadlock in a large system'' (later called Monotonicity Lemma).  In Sect.~\ref{sec:ideas-disj-deadlock} and \ref{sec:ideas-conj-deadlock} we
prove these lemmas for the case of general deadlock (global \emph{or} local).
\end{remark}

\myparagraph{Challenge: Open Systems}
For any open system $S$ there exists a closed system $S'$ such that $
S$ and $S'$ cannot be distinguished by $\LTL$ specifications 
(cp. Manna and Pnueli~\cite{Manna92}). Thus, one approach to PMC for open 
systems is to use a translation between open and closed systems, and then use the 
existing cutoff results for closed systems.

While such an approach works in theory, it might not be feasible in practice:
since cutoffs depend on the size of process templates,
and the translation blows up the process template,
it also blows up the cutoffs.
Thus, cutoffs that directly support open systems are important.

\myparagraph{Challenge: Liveness and Deadlocks under Fairness}
We are interested in cutoff results that support liveness properties.
In general, we would like to consider only runs where all processes move infinitely often, i.e., use the unconditional fairness assumption $\forall{p}. \GF \sched_p$. However, this would mean that we accept all systems that always go into a local deadlock, since then the assumption is violated. This is especially undesirable in synthesis, because the synthesizer usually tries to violate the assumptions in order to satisfy the specification. To avoid this, we require the absence of local deadlocks under the strong fairness assumption $\forall{p}. (\GF \enabled_p \impl \GF \sched_p)$. Since strong fairness and absence of local deadlocks imply unconditional fairness, we can then use the latter as an assumption for the original specification.


In summary, for a parameterized specification $\spec$, we consider satisfaction of
\[
\begin{array}{lllll}
\textit{``all runs are infinite''} &~~\land~~& \pforall_{strong} \spec_{\neg dead} & ~~\land~~ & \pforall_{uncond} \spec.
\end{array}
\]
This is equivalent to $\textit{``all runs are infinite''} \land \pforall_{strong} (\spec_{\neg dead} \,\land\, \spec)$, but by considering the form above we can separate the tasks of deadlock detection and of model checking $\LTLmX$-properties, and obtain modular cutoffs. 

In the following, we present cutoffs for problems of the forms 
(i) $\pforall_{uncond} \spec$,
(ii) $\pforall_{strong} \spec_{\neg dead}$ and no global deadlocks 
(and the variants with $\pexists$ path quantifier).

%% file: cutoffs.tex
\section{New Cutoff Results}
\label{sec:cutoffs}

We present new cutoff results that extend Theorems~\ref{thm:disj-cutoff-pairs} and \ref{thm:conj-cutoff}, summarized in the table below. We distinguish between disjunctive and conjunctive systems, non-fair and fair executions, as well as between the satisfaction of $\LTLmX$ properties $h(A,B^{(k)})$ and the existence of deadlocks. All results hold for open systems, and for both path quantifiers $\pforall$ and $\pexists$. Cutoffs depend on the size of process template $B$ and the number $k \geq 1$ of $B$-processes a property talks about:
\begin{table}[h]
\centering
\vspace{-10pt}
\label{table:cutoffs}
\centering
\setlength{\tabcolsep}{2pt}
{
\begin{tabular}{ r|c|c|c|c }
   & \specialcellC{$h(A,B^{(k)})$ \\ no fairness} & \specialcellC{deadlock detection \\ no fairness} & \specialcellC{$h(A,B^{(k)})$ \\ uncond. fairness} & \specialcellC{deadlock detection \\ strong fairness}  \\[9pt]
\hline
Disjunctive~ & $|B| + k + 1$ &
        $2|B| - 1$ & 
        $2|B| + k - 1$ &
        $2|B| - 1$ 
          \\[4pt]
\hline
Conjunctive~ & 
        $k+1$ &  
        $2|B|-2~(*)$ & 
        $k+1~(*)$ &  
        $2|B|-2~(*)$
\end{tabular}
}
\vspace{-10pt}
\end{table}

Results marked with a $(*)$ are for a restricted class of systems:
For conjunctive systems with fairness, we require infinite runs to be 
initializing, i.e., all non-deadlocked 
processes 
return to 
$\init$
infinitely often.\footnote{This assumption is in the same 
flavor as the restriction that $\initstate_A$ and $\initstate_B$ appear in 
all conjunctive guards. Intuitively, the additional restriction makes sense 
since conjunctive systems model shared resources, and everybody who takes a 
resource should eventually release it.} Additionally, the cutoffs for 
deadlock detection in conjunctive systems only support 
$1$-conjunctive systems. The reason for this restriction will be
explained in Remark~\ref{rem:general-conj-tough}\iffinal.\else (and Appendix~\ref{sec:app-conj}).\fi

All cutoffs in the table are tight -- no smaller cutoff can exist for 
this class of systems and properties -- except for the case of deadlock 
detection in disjunctive systems without fairness. There, the cutoff is 
asymptotically tight, i.e., it must increase linearly with the size 
of the process template.

\input{proof-structure}
\input{proof-ideas-disj}
\input{proof-ideas-conj}

%% file: proof-structure.tex
\subsection*{Proof Structure}
To justify the entries in the table,
we first recapitulate the proof structure of the original 
Theorems~\ref{thm:disj-cutoff-pairs} and \ref{thm:conj-cutoff}.
The proofs are based on two lemmas, Monotonicity and Bounding.
We give some basic proof ideas of the lemmas from~\cite{Emerson00} and mention extensions to the cases with fairness and deadlock detection. For cases where this extension is not easy, we will introduce additional proof techniques and explain how to use them in Sections~\ref{sec:ideas-disj} and \ref{sec:ideas-conj}. 
Note that we only consider properties of the form $h(A,B^{(1)})$ --- the
proof ideas extend to general properties $h(A,B^{(k)})$ without
difficulty. Similarly, in most cases the proof ideas extend to open systems
without major difficulties --- mainly because when we construct a simulating
run, we have the freedom to choose the input that is needed. 
Only for the case of deadlock detection we have to handle open systems explicitly.

\smallskip
\noindent
{\bf 1) \emph{Monotonicity} lemma:} if a behavior is possible in a system with $n \in \Nat$ copies of $B$, then it is also possible in a system with one additional process:
\[
\largesys \models \pexists h(A,B^{(1)}) 
~\implies~
(A,B)^{(1, n+1)} \models \pexists h(A,B^{(1)}), 
\]
and if a deadlock is possible in $(A,B)^{(1, n)}$, then it is possible in $(A,B)^{(1, n+1)}$.

\begin{proof}[Proof ideas] The lemma is easy to prove for properties 
$\pexists h(A,B^{(1)})$ in both disjunctive and conjunctive systems, by letting the 
additional process stay in its initial state $\init_B$ forever 
(cp.~\cite{Emerson00}). This cannot disable transitions with disjunctive guards, as 
these check for \emph{existence} of a local state in another process (and we 
do not remove any processes), and it cannot disable conjunctive guards since 
they contain $\init_B$ by assumption. 
However, this construction violates fairness, since the new process 
never moves. This can be resolved in the disjunctive case by letting the 
additional process mimic all transitions of an existing process. But in 
general this does not work in conjunctive systems (due to the non-reflexive 
interpretation of guards).
For this case and for deadlock detection, the proof is not 
trivial and may only work for $n \geq c$, for some lower bound $c \in \Nat$ 
(see Sect.~\ref{sec:ideas-disj}, \ref{sec:ideas-conj}).
\end{proof}

\smallskip
\noindent
{\bf 2) \emph{Bounding} lemma:} for a number $c \in \Nat$, a behavior is
possible in a system with $c$ copies of $B$ if it is possible in a system with
$n \geq c$ copies of process $B$:
\[
(A,B)^{(1, c)} \models \pexists h(A,B^{(1)})
~\impliedby~
(A,B)^{(1, n)} \models \pexists h(A,B^{(1)})
,
\]
and a deadlock is possible in \cutoffsys if it is possible in \largesys.

\begin{proof}[Proof ideas] 
  For disjunctive systems, the main difficulty is that removing processes might falsify guards of the local transitions of $A$ or $B_1$ in a given run (see Sect.~\ref{sec:ideas-disj}).
  For conjunctive systems, removing processes from a run is easy for the case of infinite runs, since a transition that was enabled before cannot become disabled. Here, the difficulty is in preserving deadlocks, because removing processes may enable processes that were deadlocked before (Sect.~\ref{sec:ideas-conj}).
\end{proof}


%% file: proof-ideas-disj.tex
\section{Proof Techniques for Disjunctive Systems}
\label{sec:ideas-disj}
\ak{to simplify the notation, can we remove `$(x)$' in $\visInf{\mB}{x}$?} \sj{I think it is much clearer if we leave it - the notation is self-explanatory currently}

\subsection{\LTLmX\ Properties without Fairness: Existing Constructions}
\label{sec:ideas-disj-nofair}

We revisit the main technique of the original proof 
of Theorem~\ref{thm:disj-cutoff-pairs}~\cite{Emerson00}. 
It constructs an infinite run $y$ of $\cutoffsys$ 
with $y \models h(A,B^{(1)})$, 
based on an infinite run $x$ of $\largesys$ with $n>c$ and $x \models h(A,B^{(1)})$. 
The idea is to copy local runs $x(A)$ and $x(B_1)$ into $y$, 
and construct runs of other processes in a way 
that enables all transitions along $x(A)$ and $x(B_1)$. 
The latter is achieved with the flooding construction.

\myparagraph{Flooding Construction \cite{Emerson00}}
Given a run $x = (s_1,e_1,p_1), (s_2,e_2,p_2) \ldots$ of $\largesys$, let
$\visited_\mB(x)$ be the set of all local states visited by $B$-processes in $x$,
i.e., $\visited_\mB(x) = \{ q \in Q_B \| \exists m \exists i.\ s_m(B_i) = q \}$. 

For every $q \in \visited_\mB(x)$ there is a local run of \largesys, say $x(B_i)$,
that visits $q$ first, say at moment $m_q$. Then, saying that process 
$B_{i_q}$ of \cutoffsys \emph{floods $q$} means:
$$y(B_{i_q}) = x(B_i)\slice{1}{m_q}(q)^\omega.$$ 
In words: the run $y(B_{i_q})$ is the same as $x(B_i)$ until moment $\time_q$,
and after that the process never moves.

The construction achieves the following. 
If we copy local runs of $A$ and $B_1$ from $x$ to $y$, 
and in $y$ for every $q \in \visited_\mB(x)$ introduce one process that floods $q$, 
then: 
if in $x$ at some moment $\time$ there is a process in state $q'$, 
then in $y$ at moment $\time$ there will also be a process (different from 
$A$ and $B_1$) in state $q'$. Thus, every transition of $A$
 and $B_1$, which is enabled at moment $\time$ in $x$, will also be enabled in $y$. 

\myparagraph{Proof idea of the bounding lemma}
The lemma for disjunctive systems without fairness can be proved by
copying local runs $x(A)$ and $x(B_1)$, and flooding all states in
$\visited_\mB(x)$. To ensure that at least one process moves infinitely often
in $y$, we copy one additional (infinite) local run from $x$. Finally, it may
happen that the resulting collection of local runs violates the interleaving 
semantics requirement. To resolve this, we add stuttering steps into local 
runs whenever two or more processes move at the same time, and we 
remove global stuttering steps in $y$. Since the only difference between 
$x(A,B_1)$ and $y(A,B_1)$ are stuttering steps, $y$ and $x$ satisfy the same $
\LTLmX$-properties $h(A,B^{(1)})$. 
Since $\card{\visited_\mB(x)} \leq 
\card{B}$, we need at most $1+\card{B}+1$ copies of $B$ in \cutoffsys.

\subsection{\LTLmX\ Properties with Fairness: New Constructions}
\label{sec:ideas-disj-fair}

The flooding construction does not preserve fairness, 
 and also cannot be used to construct deadlocked runs since it does not 
preserve  disabledness of transitions of processes $A$ or $B_1$. 
For these cases, we provide new proof constructions.

Consider the proof task of the bounding lemma for disjunctive systems with 
fairness: given an unconditionally fair run $x$ of 
\largesys with 
$x \models h(A,B^{(1)})$, we want to construct an unconditionally fair run $y$ 
of \cutoffsys with $y \models h(A,B^{(1)})$. In contrast to unfair systems, we 
need to ensure that all processes move infinitely often in $y$. 
The insight is 
that after a finite time all processes will start looping 
around some set $\visited^\inf$ of states. We construct a run $y$ that
mimics this. To this end, we introduce two constructions. \emph{Flooding with
evacuation} is similar to flooding, but instead of keeping
processes in their flooding states forever it evacuates the processes into 
$\visited^\inf$. \emph{Fair extension} lets all processes move infinitely 
often without leaving $\visited^\inf$.

\myparagraph{Flooding with Evacuation}
\input{flooding-with-evacuation}

\myparagraph{Fair Extension} 
\ak{explain intuition about those three sets}
\input{fair_extension}
\myparagraph{Proof idea of the bounding lemma}
\input{proof-disj-bounding-fair}

\subsection{Detection of Local and Global Deadlocks: New Constructions}
\label{sec:ideas-disj-deadlock}

\myparagraphraw{Monotonicity Lemmas.}
The lemma for deadlock detection, for fair and unfair cases,
is proven for $n \geq |B|+1$.
In the case of local deadlocks, 
process $B_{n+1}$ mimics a process that moves infinitely often in $x$.
In the case of global deadlocks, 
by pigeon hole principle, 
in the global deadlock state there is a state $q$ with at least two processes in it---let process $B_{n+1}$ mimic a process that deadlocks in $q$.

\myparagraphraw{Bounding Lemmas.}
For the case of global deadlocks, fairness does not affect the proof of the bounding lemma. 
The insight is to divide deadlocked local states into two disjoint sets, 
$\dead_1$ and $\dead_2$, as follows.
Given a globally deadlocked run $x$ of \largesys, 
for every $q \in \dead_1$, 
there is a process of \largesys deadlocked in $q$ with input $i$,
that has an outgoing transition guarded ``$\exists q$''
-- hence, adding one more process into $q$ would unlock the process.
In contrast, $q \in \dead_2$ if any process deadlocked in $q$
stays deadlocked after adding more processes into $q$.
Let us denote the set of $B$-processes deadlocked in $\dead_1$ by $\mD_1$.
Finally, abuse the definition in Eq.~\ref{disj:def_vfin_wrt}
and denote by $\visFin{\mB\smi\mD_1}{x}$ the set of states
that are visited by $B$-processes not in $\mD_1$ before reaching a deadlocked state.

Given a globally deadlocked run $x$ of \largesys with $n\geq 2|B|-1$, 
we construct a globally deadlocked run $y$ of \cutoffsys with $c = 2|B|-1$ as follows:
\li
\- copy from $x$ into $y$ the local runs of processes in $\mD_1 \cup \{A\}$
\- flood every state of $\dead_2$
\- for every $q \in \visFin{\mB\smi\mD_1}{x}$, flood $q$ and evacuate into $\dead_2$.
\il
The construction ensures: 
(1) for any moment and any process in $y$,
    the set of local states that are visible to the process includes all the states that were visible 
    to the corresponding process in \largesys whose transitions we copy;
(2) in $y$, there is a moment when all processes deadlock in $\dead_1 \cup \dead_2$.

For the case of local deadlocks, 
the construction is similar but slightly more involved, 
and needs to distinguish between unfair and fair cases.
In the unfair case, we also copy the behaviour of an infinitely moving process. 
In the strong-fair case,
we continue the runs of non-deadlocked processes with the fair extension. 
\iffinal \else See details in Appendix~\ref{sec:app-disj}.\fi

\ak{put here the tightness picture for deadlocks under fairness?}

%
%

%% file: flooding-with-evacuation.tex
Given a subset $\mF \subseteq \mB$
and an infinite run $x=(s_1,e_1,p_1)\ldots$ of \largesys, 
define
\begin{align}
& \visInf{\mF}{x} = \{ q \|\! \exists \text{ infinitely many } ~~~~ m\!:  
s_m(B_i) = q 
\text{ for some } B_i \in \mF \} \label{disj:def_vinf_wrt} \\
& \visFin{\mF}{x} = \{ q \|\! \exists \text{ only finitely many } m\!:  
s_m(B_i) = q
\text{ for some } B_i\in \mF \} \label{disj:def_vfin_wrt}
\end{align}
Let $q \in \visFin{\mF}{x}$.
In run $x$ there is a moment $f_q$ when $q$
is reached for the first time by some process from $\mF$, denoted $B_{\first_q}$. 
Also, in run $x$ there is a moment $l_q$ such that:
$s_{l_q}(B_{\last_q})=q$ for some process $B_{\last_q} \in \mF$, 
and $s_t(B_i)\neq q$ for all $B_i \in \mF$,
$t > l_q$ 
--- i.e., when some process from $\mF$ is in state $q$ for the last time in $x$. 
Then, saying that process $B_{i_q}$ of \cutoffsys 
\emph{floods $q \in \visFin{\mF}{x}$ and then evacuates into $\visInf{\mF}{x}$} 
means: 
$$
y(B_{i_q}) = x(B_{\first_q})\slice{1}{f_q} \ \cdot\ (q)^{(l_q - f_q + 1)} \cdot \ 
x(B_{\last_q})\slice{l_q}{m} \ \cdot \ (q')^\omega,
$$
where $q'$ is the state in $\visInf{\mF}{x}$ that $x(B_{\last_q})$ reaches first, 
at some moment $\time \geq l_q$.
In words, process $B_{i_q}$ mimics process $B_{\first_q}$ until it reaches $q$, 
then does nothing until process $B_{\last_q}$ starts leaving $q$, 
then it mimics $B_{\last_q}$ until it reaches $\visInf{\mF}{x}$.

The construction ensures: 
if we copy local runs of all processes not in $\mF$ from $x$ to $y$, 
then all transitions of $y$ are enabled. 
This is because: 
for any process $p$ of $\cutoffsys$ that takes a transition in $y$ at any moment, 
the set of states visible to process $p$ is a superset of the set of states 
visible to the original process in \largesys whose transitions process $p$ copies.

%% file: fair_extension.tex
\ak{adapt to dead}
Here, we consider a path $x$ that is the postfix of an unconditionally fair run $x'$ of $\largesys$, 
starting from the moment where no local states from $\visFin{\mB}{x'}$ are visited anymore. 
We construct a corresponding unconditionally-fair path $y$ of $\cutoffsys$, 
where no local states from $\visFin{\mB}{x'}$ are visited.

Formally, let $n \geq 2|B|$, and $x$ an unconditionally-fair path of $\largesys$ such that
$\visFin{\mB}{x}=\emptyset$.
Let $c \geq 2|B|$, and $s_1'$ a state of \cutoffsys
with
\li
\- $s_1'(A_1)=s_1(A_1)$, $s_1'(B_1)=s_1(B_1)$

\- for every $q \in \visInf{B_2..B_n}{x} \smi \visInf{B_1}{x}$,
   there are two processes $B_{i_q}, B_{i_q'}$ of \cutoffsys
   that start in $q$, i.e., $s_1'(B_{i_q})=s_1'(B_{i_q'})=q$

\- for every $q \in \visInf{B_2..B_n}{x} \cap \visInf{B_1}{x}$,
   there is one process $B_{i_q}$ of \cutoffsys
   that starts in $q$

\- for some $\qstar \in \visInf{B_2..B_n}{x} \cap \visInf{B_1}{x}$,
   there is one additional process of \cutoffsys, 
   different from any in the above, 
   called $B_{i_\qstar'}$,
   that starts in $\qstar$.

\- any other process $B_i$ of \cutoffsys 
   starts in some state of $\visInf{B_2..B_n}{x}$.
\il
Note that if $\visInf{B_2..B_n}{x}\cap \visInf{B_1}{x} = \emptyset$, 
then the third and fourth pre-requisites are trivially satisfied.

The fair extension extends state $s_1'$ of \cutoffsys 
to an unconditionally-fair path $y=(s'_1,e'_1,p'_1)\ldots$ 
with $y(A_1,B_1) = x(A_1,B_1)$ as follows:
\li
\-[(a)] $y(A_1)=x(A_1)$, $y(B_1)=x(B_1)$

\-[(b)] for every $q \in \visInf{B_2..B_n}{x} \smi \visInf{B_1}{x}$: 
       in run $x$ there is $B_i \in \{B_2..B_n\}$ 
       that starts in $q$ and visits it infinitely often. 
       Let $B_{i_q}$ and $B_{i'_q}$ of \cutoffsys mimic $B_i$ in turns: 
       first $B_{i_q}$ mimics $B_i$ until it reaches $q$, 
       then $B_{i'_q}$ mimics $B_i$ until it reaches $q$, and so on.

\-[(c)] arrange states of $\visInf{B_2..B_n}{x}\cap \visInf{B_1}{x}$ 
       in some order $(\qstar, q_1, \ldots, q_l)$.  
       The processes $B_{i_\qstar'}, B_{i_\qstar}, B_{i_{q_1}}, \ldots, B_{i_{q_l}}$ 
       behave as follows.
       Start with $B_{i_\qstar'}$: 
       when $B_1$ enters $\qstar$ in $y$, it carries%
       \footnote{``Process $B_1$ starting at moment $m$ carries process $B_i$ 
                 from $q$ to $q'$'' means: process $B_i$ mimics 
                 the transitions of $B_1$ starting at moment $m$ at $q$ 
                 until $B_1$ first reaches $q'$.}
       $B_{i_\qstar'}$             from $\qstar$ to $q_1$, 
       then carries $B_{i_{q_1}}$ from $q_1$ to $q_2$, \ldots, 
       then carries $B_{i_{q_l}}$ from $q_l$ to $\qstar$, 
       then carries $B_{i_\qstar}$ from $\qstar$ to $q_1$, 
       then carries $B_{i_\qstar'}$ from $q_1$ to $q_2$, 
       then carries $B_{i_{q_1}}$ from $q_2$ to $q_3$,
       and so on.


\-[(d)] any other $B_i$ of \cutoffsys,
       starting in $q \in \visInf{B_2..B_n}{x}$,
       mimics $B_{i_q}$.
\il
Note that parts (b) and (c) of the constrution ensure that there is always at
       least one process in every state from $\visInf{B_2..B_n}{x}$. This
       ensures that the guards of all transitions of the construction are satisfied.
Excluding processes in (d), the fair extension uses up to $2|B|$ copies of $B$.%
\footnote{A careful reader may notice that if
          $|\visInf{B_1}{x}|=1$ and $|\visInf{B_2..B_n}{x}|=|B|$,
          then the construction uses $2|B|+1$ copies of $B$.
          But one can slightly modify the construction for this special case,
          and remove process $B_{i_\qstar'}$ from the pre-requisites.}

%% file: proof-disj-bounding-fair.tex
\sj{for weak or strong fairness, the same construction can be used; evacuation is not necessary, but also doesn't increase the cutoff if we use it; difficulty: show that cutoff is still tight
}
Let $c=2\card{B}$. 
Given an unconditionally-fair run $x$ of $\largesys$ 
we construct an unconditionally-fair run $y$ of the cutoff system $\cutoffsys$ 
such that $y(A,B_1)$ is stuttering equivalent to $x(A,B_1)$.

Note that in $x$ there is a moment $m$ such that all local states that are visited after $m$ are in $\visInf{\mB}{x}$.

The construction has two phases. In the first phase, we apply flooding for states in $\visInf{\mB}{x}$, and flooding with evacuation for states in $\visFin{\mB}{x}$:
\li
\-[(a)] $y(A)=x(A)$, $y(B_1)=x(B_1)$

\-[(b)] for every $q \in \visInf{B_2..B_n}{x} \smi \visInf{B_1}{x}$, 
       devote two processes of $\cutoffsys$ that flood $q$

\-[(c)] for some $\qstar \in \visInf{B_2..B_n}{x} \cap \visInf{B_1}{x}$,
       devote one process of \cutoffsys that floods $\qstar$

\-[(d)] for every $q \in \visFin{B_2..B_n}{x}$, 
       devote one process of $\cutoffsys$ that 
       floods $q$ and evacuates into $\visInf{B_2..B_n}{x}$

\-[(e)] let other processes (if any) mimic process $B_1$
\il
The phase ensures that at moment $m$ in $y$, 
there are no processes in $\visFin{\mB}{x}$, 
and all the pre-requisites of the fair extension are satisfied.

The second phase applies the fair extension, 
and then establishes the interleaving semantics 
as in the bounding lemma in the non-fair case.
The overall construction uses up to $2|B|$ copies of $B$.

%% file: proof-ideas-conj.tex
\section{Proof Techniques for Conjunctive Systems}
\label{sec:ideas-conj}

\subsection{\LTLmX\ Properties without Fairness: Existing Constructions}
\label{sec:ideas-conj-nofair}
Recall that the Monotonicity lemma is proven by keeping the additional process
in the initial state.
To prove the bounding lemma, 
Emerson and Kahlon \cite{Emerson00} suggest to simply copy the local runs $x(A)$ and $x(B_1)$ into $y$. 
In addition, we may need one more process that moves infinitely often to ensure that an infinite run of \largesys will result in an infinite run of \cutoffsys. All transitions of copied processes will be enabled because removing processes from a conjunctive system cannot disable a transition that was enabled before.

\subsection{\LTLmX\ Properties with Fairness: New Constructions}
\label{sec:ideas-conj-fair}

The proof of the 
Bounding lemma is the same as in the non-fair case, noting that if the 
original run is unconditional-fair, then so will be the resulting run.

Proving the Monotonicity lemma is more difficult, since the fair extension 
construction from disjunctive 
systems does not work for conjunctive systems
-- if an additional process mimics the transitions of an existing process
   then it disables transitions of the form 
   $\transition{q}{q'}{\textit{``\,}\forall\neg q\textit{\!''}}$ or
   $\transition{q}{q'}{\textit{``\,}\forall\neg q'\textit{\!''}}$.
Hence, we add the restriction of initializing runs, 
which allows us to construct a fair run as follows.
The additional process $B_{n+1}$ ``shares'' a local run $x(B_i)$ 
with an existing process $B_i$ of $(A,B)^{(1,n+1)}$: 
one process stutters in $\init_B$ while the other makes transitions from $x(B_i)$, 
and whenever $x(B_i)$ enters $\init_B$
(this happens infinitely often),
the roles are reversed. 
Since this changes the behavior of $B_i$, 
$B_i$ should not be mentioned in the formula, 
i.e., we need $n\geq 2$ for a formula $h(A,B^{(1)})$.

\subsection{Detection of Local and Global Deadlocks: New Constructions}
\label{sec:ideas-conj-deadlock}
\myparagraphraw{Monotonicity lemmas} for both fair and unfair cases 
are proven by keeping process $B_{n+1}$ in the initial state, 
and copying the runs of deadlocked processes.
If the run of \largesys is globally deadlocked, then process $B_{n+1}$ may keep moving in the constructed run, i.e., it may only be locally deadlocked. 
In case of a local deadlock in \largesys, distinguish two cases: 
there is an infinitely moving $B$-process, or all $B$-processes are deadlocked 
(and thus $A$ moves infinitely often).
In the latter case, use the same construction as in the global deadlock case
(the correctness argument uses the fact that systems are 1-conjunctive, 
 runs are initializing, and there is only one process of type $A$).
In the former case, copy the original run, and let $B_{n+1}$ share
a local run with an infinitely moving $B$-process.

\myparagraphraw{Bounding lemma (no fairness).}
In the case of global deadlock detection, 
Emerson and Kahlon~\cite{Emerson00} suggest to copy a subset of the original local runs.
For every local state $q$ that is present in the final state of the run, 
we need at most two local runs that end in this state. 
In the case of local deadlocks, 
our construction uses the fact that systems are 1-conjunctive.
In 1-conjunctive systems, if a process is deadlocked, 
then there is a set of states $DeadGuards$ that all need to be populated by other processes
in order to disable all transitions of the deadlocked process. 
Thus, 
the construction copies: 
(i) the local run of a deadlocked process, 
(ii) for each $q \in DeadGuards$, the local run of a process 
     that is in $q$ at the moment of the deadlock, and
(iii) the local run of an infinitely moving process.
\iffinal \else See the details in Appendix~\ref{sec:app-conj}.\fi

\myparagraphraw{Bounding lemma (strong fairness).}
We use a construction that is similar to that of properties under fairness for disjunctive systems (Sect.~\ref{sec:ideas-disj-fair}):
  in the setup phase, 
  we populate some ``safe'' set of states with processes,
  and then we extend the runs of non-deadlocked processes 
  to satisfy strong fairness, 
  while ensuring that deadlocked processes never get enabled.

Let $c=2|Q_B\smi \{ \init_B \}|$. 
Let $x= (s_1,e_1,p_1)\ldots$ be a locally deadlocked strong-fair intitializing run 
of $\largesys$ with $n>c$. 
We construct a locally deadlocked strong-fair initializing run $y$ of $\cutoffsys$.

Let $\mD \subseteq \mB$ be the set of deadlocked $B$-processes in $x$. 
Let $d$ be the moment in $x$ starting from which every process in $\mD$ is deadlocked.
Let $\dead(x)$ be the set of states in which processes $\mD$ of \largesys
are deadlocked.
Let $\dead_2(x) \subseteq \dead(x)$ be the set of deadlocked states such that: 
for every $q \in \dead_2(x)$, 
there is a process $B_i \in \mD$ with $s_d(B_i) = q$ 
and that for input $e_{\geq d}(B_i)$ has a transition guarded with ``$\forall \neg q$''.
Thus, a process in $q$ is deadlocked with $e_d(B_i)$
only if there is another process in $q$ in every moment $\geq d$.
Let $\dead_1(x) = \dead(x)\smi\dead_2(x)$.
Define $DeadGuards$ to be the set
$$
\{\ q \| \exists B_i \in \mD
         \textit{ with a transition guarded ``\,}
         {\forall \neg q} 
         \textit{\!'' in } (s_d(B_i),e_d(B_i))\ \}.
$$ 
Figure~\ref{fig:ideas:conj-deadlocks-venn} illustrates properties of sets 
$DeadGuards$, $\dead_1$, $\dead_2$, $\visInf{\mB\smi\mD}{x}$.
\begin{figure}[t]
\vspace{-0.4cm}
\centering
\includegraphics[width=0.7\textwidth]{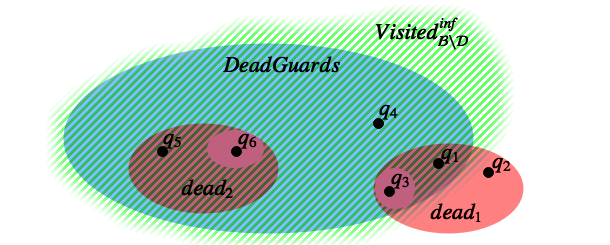}
\caption[fig:ideas:conj-deadlocks-venn]{%
Bounding lemma (strong fairness):
Venn diagram for 
$\dead_1$, $\dead_2$, $DeadGuards$, $\visInf{\mB\smi\mD}{x}$.
States $q_1,...,q_6$ are to illustrate that the corresponding sets may be non-empty.
E.g., 
in $x$, 
a process may be deadlocked in $q_1 \in (DeadGuards \cap \dead_1 \cap \visInf{\mB\smi\mD}{x})$, 
and another process in $q_3 \in \dead_1 \cap DeadGuards \smi \visInf{\mB\smi\mD}{x}$.%
}
\label{fig:ideas:conj-deadlocks-venn}
\vspace{-0.4cm}
\end{figure}

In the {\bf setup phase}, we copy from $x$ into $y$: 
\li
\- the local run of $A$; 
\- for every $q \in \dead_1$, the local run of one process deadlocked in $q$;
\- for every $q\in \dead_2$, the local runs of two%
   \footnote{Strictly speaking, 
             in $x$ we might not have two deadlocked processes 
             in a state in $dead_2$
             -- one process may be deadlocked, 
                others enter and exit the state infinitely often.
             In such case, there is always a non-deadlocked process in the state.
             Then, copy the local run of such infinitely moving process 
             until it enters the deadlocked state, and then deadlock it 
             by providing the same input as the deadlocked process receives.} 
   processes deadlocked in $q$;
\- for every $q \in DeadGuards \smi \dead$, the local run of a process that reaches $q$ after moment $d$.
\- Finally, we keep one $B$-process in $\init_B$ until moment $d$.  
\il
The setup phase ensures: 
in every state $q \in \dead$,
there is at least one process deadlocked in $q$ at moment $d$ in $y$. 
Now we need to ensure that the non-deadlocked processes in $DeadGuards\smi \dead$ and $\init_B$
move infinitely often, which is done using the looping extension described bellow.

Order arbitrarily 
$DeadGuards \smi \dead = (q_1,\ldots,q_k) \subseteq \visInf{\mB\smi\mD}{x}$.
Let $\mP \subseteq \{B_1,...,B_c\}$ be the non-deadlocked processes of \cutoffsys that we moved 
into $(q_1,\ldots,q_k) \cupdot \{\init_B\}$ in the setup phase.
Note that $|\mP| = |(q_1,...,q_k)| + 1$.

The {\bf looping phase} is: set $i=1$, and repeat infinitely the following.
\li
  \- let $B_\init \in \mP$ be the process of \cutoffsys that is currently in $\init_B$, 
     and $B_{q_i} \in \mP$ be the process of \cutoffsys that is currently in $q_i$
     
  \- let $\tilde{B}_{q_i} \in \visInf{\mB\smi\mD}{x}$ be a process of \largesys 
     that visits $q_i$ and $\init_B$ infinitely often.
     Let $B_\init$ of \cutoffsys copy transitions of $\tilde{B}_{q_i}$
     on some path $\init_B \to \ldots \to q_i$,
     then let $B_{q_i}$ copy transitions of $\tilde{B}_{q_i}$ on some path 
     $q_i \to \ldots \to \init_B$. 
     For copying we consider only the paths of $\tilde{B}_{q_i}$ that happen after moment $d$.

  \- $i=i \oplus 1$
\il
\iffinal \else For more details, including tightness observations, see Appendix~\ref{sec:app-conj}.\fi
\begin{remark}\label{rem:general-conj-tough}
In 1-conjunctive systems, the set $DeadGuards$ is ``static'', i.e.,
  there is always at least one process in \emph{each state} of $DeadGuards$ 
  starting from the moment of the deadlock.
In contrast, in general conjunctive systems where guards can overlap, 
  there is no such set. However, there is a similar set of sets of states, such that \emph{one state from each set} always needs to be populated to ensure the deadlock. 
\end{remark}
  
%
%

%% file: experiments.tex
\section{Experiments}
\label{sec:experiments}

We used our new cutoff results in the context of parameterized synthesis to 
automatically construct process templates with safety and lifeness guarantees 
for guarded systems with an arbitrary number of components.
Our prototype is an extension of the parameterized synthesis tool 
{\sc Party}~\cite{Khalimov13a}.
It synthesizes process templates based on the semi-decision procedure 
described in Section~\ref{sec:paramsynt}. 
Bounded template synthesis is 
implemented as an extension of the bounded synthesis 
approach~\cite{FinkbeinerS13}, with
LTL3BA~\cite{BabiakKRS12} for translation of specifications to automata, and
SMT solver Z3~\cite{DeMour08} as backend. The method supports additional state labels that can be used as output 
signals in 
specifications. Cutoffs are detected automatically and applied modularly for 
different properties and deadlock detection (cp. \cite{KhalimovJB13b}).

%
%
As a 
proof of concept, we synthesized implementations for two small 
examples (on a single core of a 3.5GHz Intel i7 CPU with 4GB RAM):
\begin{itemize}
\item a closed disjunctive system $\largesys$ where $A$ controls an output 
signal $w$, and every $B_i$ a signal $g_i$, with fairness assumption and specification
$\GF(\neg w) \wedge \GF w \wedge 
\bigwedge_{i} \left[\GF(w \wedge g_i) \wedge  \GF g_i \wedge \GF(\neg g_i) 
\right]$,
i.e., all processes must toggle their signals 
infinitely often, and every process $B^i$ must ``meet'' infinitely often with 
$A$ when both signals are enabled. The implementation below has been synthesized with a cutoff of $2\card{B}+1-1=4$ for the property and $2\card{B}-1=3$ for the deadlock detection within one minute.

\item an open conjunctive system $(B)^{(n)}$ (i.e., $A$ can be arbitrary and $B$ does not react to it), where each process $B_i$ gets requests of high and low priority as input ($rh_i, rl_i$), and controls ouputs that represent grants to these requests ($gh_i, gl_i$). The specification requires i) `every request $rh$ should eventually be granted ($gh_i$)', ii) `every request $rl_i$ should eventually be granted ($gl_i$) unless there is a simultaneous request $rh_i$', iii) `there should be no spurious grants', iv) `mutual exclusion of grants (locally and globally)'. The implementation above has been synthesized with a cutoff of $3$ for properties, and $4$ for deadlock detection within about $45$ minutes. Note that we do not restrict the implementation to $1$-conjunctive systems, and therefore the cutoff does not guarantee absence of local deadlocks in systems of arbitrary size.
\end{itemize}

\begin{figure}
\vspace{-20pt}
\centering
\subfloat[Disjunctive Implementation]{
\centering
\scalebox{0.75}{\input{img/disjunctive_lts}}
\label{fig:disjunctiveLTS}
}
\hspace{1cm}
\subfloat[Conjunctive Implementation]{
\centering
\scalebox{0.75}{\input{img/mutual_exclusion3_lts}}
\label{fig:conjunctiveLTS3}
}
\label{fig:experiments}
\end{figure}
\vspace{-20pt}
\sj{figure misplaced, no caption}

%% file: img/disjunctive_lts.tex
\begin{tikzpicture}[->,>=stealth',shorten >=1pt,auto,node distance=1.6cm,
                    semithick]
  \tikzstyle{every state}=[align=center,anchor=center]
  \tikzstyle{every edge} = [align=center,draw=black]
  \tikzstyle{boxLabel} = [yshift=-0.4cm]
  \tikzstyle{box} = [draw=black, inner sep=0.75cm]
	
  \node[state] (t00) [label={left:$\neg w$}]                    {$0_A$};
  \node[state] (t01) [below of=t00,label={left:$w$}] {$1_A$};
  \coordinate (t00') at ($(t00.west) + (-1.0cm,0.5cm)$);
  \coordinate (t01') at ($(t01.east) + (+0.75cm,-0.25cm)$);
  \node (U1) [draw=black, fit=(t00') (t01'), inner sep=0.75cm] {} ;
  \node [boxLabel] at (U1.north) {Template $A$};
  
  \path (t00) edge [bend right,left] node {$\exists \left\{0_B, 1_B\right\}$} (t01)
            (t01) edge [bend right, right] node {$\exists \left\{1_B\right\}$} (t00);

  \node[state] (t10) [label={left:$\neg g$},right of=t00,node distance=4.5cm] {$0_B$};
  \node[state] (t11) [label={left:$g$}, below of=t10] {$1_B$};
  \coordinate (t10') at ($(t10.west) + (-0.75cm,0.5cm)$);
  \coordinate (t11') at ($(t11.east) + (+0.75cm,-0.25cm)$);
  \node (U2) [box,fit=(t10') (t11')] {} ;
  \node [boxLabel] at (U2.north) {Template $B$};
  
  \path (t10) edge [bend right,left] node {$\exists \left\{1_A\right\}$} (t11)
            (t11) edge [bend right, right] node {$\exists \left\{1_A\right\}$} (t10);
\end{tikzpicture}

%% file: img/mutual_exclusion3_lts.tex
\begin{tikzpicture}[->,>=stealth',shorten >=1pt,auto,node distance=1.5cm,
                    semithick]
  \tikzstyle{every state}=[align=center, anchor=center]
  \tikzstyle{every edge} = [align=center,draw=black,]

  \node[state] (A) [label={below:$$}] {$0_A$};
  \node[] (BC) [below of=A] {$$};
  \node[state] (B) [left of=BC, label={below:$gh$}] {$1_A$};
  \node[state] (C) [right of=BC, label={below:$gl$},] {$2_A$};

  \path 
  (A) edge [loop above] node {$\neg rh \wedge \neg rl : \forall \{0_A\}$} (A)
  (A) edge [left, bend right=10] node {$rh: \forall \{0_A\}$} (B)
  (B) edge [right, pos=0.3, bend right=10] node {$*: \forall \{0_A\}$} (A)
  (A) edge [right, bend left=10] node {$\neg rh \wedge rl: \forall \{0_A\}$} (C)
  (C) edge [right, bend left=10] node {$$} (A);
\end{tikzpicture}

%% file: conclusion.tex
\section{Conclusion}
\label{sec:concl}

We have extended the cutoff results for guarded protocols of Emerson and 
Kahlon~\cite{Emerson00} to support local deadlock detection, fairness 
assumptions, and open systems. In particular, our results imply decidability of the parameterized model checking problem for this class of systems and specifications, which to the best of our knowledge was unknown before. 
Furthermore, the cutoff results can easily be integrated into 
the parameterized synthesis approach~\cite{Jacobs14,Khalimov13a}.


Since conjunctive guards can model atomic sections and read-write locks, 
and disjunctive guards can model pairwise rendezvous 
(for some classes of specifications, cp.~\cite{EmersonK03}), 
our results apply to a wide spectrum of systems models.
But the expressivity of the model 
comes at a high cost: cutoffs are linear in the size of a process, and 
are shown to be tight (with respect to this parameter).
For conjunctive systems, our new results are restricted to systems with
1-conjunctive guards, effectively only allowing to model a single shared
resource. 
We conjecture that our proof methods can be extended to systems with
more general conjunctive guards, at the price of even bigger cutoffs.
We leave this extension and the question of finding cutoffs that are independent of the size of processes for future research.
%
%
\ak{note that EK have better complexities for 'for all paths' properties. 
As a future work, one can look if our cutoffs can be improved.}

\begingroup
\footnotesize
\smallskip\noindent\textbf{Acknowledgment.}
We thank Roderick Bloem, Markus Rabe and Leander Tentrup for comments on drafts of this paper.
This work was supported by the Austrian Science Fund (FWF) through the 
RiSE project (S11406-N23, S11407-N23) and grant nr.~P23499-N23,
as well as by the German Research Foundation (DFG) through SFB/TR 14 AVACS and
project ASDPS (JA 2357/2-1).
%

%% file: appendix.tex
\newpage
\ak{change $B_q$ to $B_{i_q}$}
\input{additional-notation}
\input{appendix-disj}
\input{appendix-conj}


%% file: additional-notation.tex
\section{Additional Definitions and Notation}
\ak{
\li
  \- check again that all proofs work when account for inputs -- some state sets that i thought are disjoit, may be not disjoint
\il
}

For a global state $s$ of system $\largesys$ and a local state $q$ (of template $A$ or $B$), 
we write $q \in s$ as shorthand for $\exists{p \in \{A,B_1,..,B_n\}} s(p) = q$.\ak{not used?} 

For a sequence $x=x_1,x_2,\ldots$ denote the subsequence between the $i$th and $j$th element 
of the sequence as $x\slice{i}{j}=x_i,\ldots,x_j$.

By $\transition{q_i}{q_j}{e:g}$ denote a process transition from $q_i$ to $q_j$ 
for input $e$ and guarded by guard $g$. 
We skip the input $e$ and guard $g$ if they are not important or can be inferred from the context.

Given system state $s$, let $Set(s)$ be the set $ \{ q \| \exists p: s(p)=q \}$.

%% file: appendix-disj.tex
\section{Cutoffs for Disjunctive Systems}
\label{sec:app-disj}

\subsection{Disjunctive Systems without Fairness}

\begin{restatable}[Monotonicity: Disj, Properties, Unfair]{lem}{DisjMonoLemma}
\label{disj:le:NonFairDisjunctiveMono}
    For disjunctive systems:
    \begin{align*}
    &\forall n \geq 1:\\
    &(A,B)^{(1,n)} \models \pexists h(A,B_1)
    \ \Impl \
    (A,B)^{(1,n+1)} \models \pexists h(A,B_1).
    \end{align*}
\end{restatable}
\begin{proof}
Given run $x$ of $(A,B)^{(1,n)}$ we construct a run $y$ of $(A,B)^{(1,n+1)}$: 
copy $x$ into $y$ and keep the additional process in the initial state.
\end{proof}

\begin{restatable}[Bounding: Disj, Properties, Unfair]{lem}{DisjBoundingLemma}
\label{disj:le:NonFairDisjunctiveBounding}
    For disjunctive systems:
    \begin{align*}
    \forall n \geq |B|+2:\ 
    (A,B)^{(1,|B|+2)} \models \pexists h(A,B_1)
    \ \ \Implied \ \ 
    \largesys \models \pexists h(A,B_1).
    \end{align*}
\end{restatable}
\noindent 
The proof is from \cite[Lemma 4.1.2]{Emerson00}. 
We recapitulate it to introduce the notion of ``a process floods a state'', 
\destutter, \interleave, and ``process mimics another process'' 
which are used in our proofs later.
\ak{we could remove the proof and define these notions outside as i did for other constructions}

\input{proof-disj-bounding-new}

\begin{restatable}[Disj, Props, Unfair]{tightness}{TightDisjBoundingLemma}
\label{obs:disj:tight_prop}
    The cutoff in Lemma~\ref{disj:le:NonFairDisjunctiveBounding} is tight, 
    i.e., for any $k$ there exist process templates $(A,B)$ with $|B| = k$ 
    and $\LTLmX$ formula $h(A,B_1)$ such that:
    $$
    (A,B)^{(1,|B|+2)} \models \pexists h(A,B_1) ~~and~~ 
    (A,B)^{(1,|B|+1)} \not\models \pexists h(A,B_1).
    $$
\end{restatable}
\input{tight-disj-bounding}

\ifwithextensions

Now, consider properties $h(A,B^{(n)})$ that may talk about ($A$ and) $n$ different copies of process template $B$. 

\begin{restatable}[Generalized Bounding Lemma]{lem}{DisjBoundingLemmaGeneral}
\label{disj:le:NonFairDisjunctiveBoundingGeneral}
For disjunctive systems without fairness:
\begin{align*}
&(A,B)^{(1,\geq \card{B} + k + 1)} \models \pexists h(A,B^{(k)}) &
&~\iff~& &
(A,B)^{(1,\card{B} + k + 1)} \models \pexists h(A,B^{(k)}).
\end{align*}
\end{restatable}
\input{proof-disj-bounding-general}

\begin{restatable}{obs}{TightDisjBoundingLemmaGeneral}
\label{obs:disj:tight_prop:general}
The cutoffs in Lemma~\ref{disj:le:NonFairDisjunctiveBoundingGeneral} are tight, i.e.,
for any $k$ there exist $(A,B)$ with $\card{B}=k$ and $\LTLmX$ formula $h(A,B^{(n)})$ such that: \sj{do we have example for generalized form?}\ak{agree, proof is needed}
$$
(A,B)^{(1,k+n+1))} \models h(A,B^{(n)}) ~~and~~ 
(A,B)^{(1,k+n)} \not\models h(A,B^{(n)}).
$$
\end{restatable}
\input{tight-disj-bounding-general}
\fi  

\begin{restatable}[Monotonicity: Disj, Deadlocks, Unfair]{lem}{mono_lem_disj_deadlocks_unfair}
\label{mono_lem_disj_deadlocks_unfair}
    For disjunctive systems:
    $$\forall n\geq |B|+1: (A,B)^{(1,n)} \textit{ has a deadlock} \ 
    \Impl\ 
    (A,B)^{(1,n+1)} \textit{ has a deadlock}$$
\end{restatable}
\begin{proof}
Given a deadlocked run $x$ of $(A,B)^{(1,n)}$ 
we build a deadlocked run of $(A,B)^{(1,n+1)}$. 
If the run $x$ is locally deadlocked,
then it has at least one infinitely moving process, 
thus let the additional process mimic that process. 
If the run $x$ is globally deadlocked run, 
then due to $n>|B|$ in some state there are at least two processes deadlocked. 
Thus, let the new process mimic a process deadlocked in that state -- 
the run constructed will also be globally deadlocked.
\end{proof}

\sj{should be tight: If we only have $\card{B}$, here is the counterexample: pipeline where transition $(q_i,q_{i+1})$ is guarded with $\exists q_i$, and $(q_{n-1},q_0)$ guarded with $\exists q_{n-1}$. This system has deadlock with up to $n$ processes, but not with $n+1$.}

\begin{restatable}[Bounding: Disj, Deadlocks, Unfair]{lem}{lem_disj_deadlocks_unfair}
\label{lem_disj_deadlocks_unfair}
For disjunctive systems:
\li
  \- with $c=|B|+2$ and any $n>c$:
  $$(A,B)^{(1,c)} \textit{ has a local deadlock} \ \Implied\ (A,B)^{(1,n)} \textit{ has a local deadlock}$$
  
  \- with $c=2|B| - 1$ and any $n>c$
  $$(A,B)^{(1,c)} \textit{ has a global deadlock} \ \Implied\ (A,B)^{(1,n)} \textit{ has a global deadlock} $$
  
  \- with $c=2|B|-1$ and any $n>c$:
  $$(A,B)^{(1,c)} \textit{ has a deadlock} \ \Implied\ (A,B)^{(1,n)} \textit{ has a deadlock}$$
\il
\ak{seems not tight}
\end{restatable}
%
%
\begin{proof}
\input{proof-disj-deadlock}

\end{proof}

\begin{restatable}[Disj, Deadlocks, Unfair]{tightness}{TightDisjDeadlockLemma}
\label{obs:disj:tight_deadlock}
The cutoff $c=2|B|-1$ for deadlock detection in disjunctive systems is \emph{asymptotically optimal but possibly not tight}, i.e.: for any $k$ there are templates $(A,B)$ with $|B|=k$ such that:
$$
(A,B)^{(1,|B|-1)} \textit{ does not have a deadlock, but } (A,B)^{(1,|B|)} \textit { does}.
$$
\end{restatable}

\input{tight-disj-deadlock}

\subsection{Disjunctive Systems with Fairness}

\begin{restatable}[Monotonicity: Disj, Props, Fair]{lem}{DisjBoundingLemmaFair}
\label{disj:le:FairDisjunctiveBounding}
For disjunctive systems:
\begin{align*}
& \forall n \geq 1: \\
&(A,B)^{(1, n)} \models \pexists_{uncond} h(A,B_1) 
\implies
(A,B)^{(1,n+1)} \models \pexists_{uncond} h(A,B_1),\\
\end{align*}
\end{restatable}
\begin{proof}
In run $x$ of $(A,B)^{(1,n)}$ with $n \geq 1$ all processes move infinitely often. 
Hence let the run $y$ of $(A,B)^{(1,n+1)}$ copy $x$, 
and let the new process mimic an infinitely moving B process of $(A,B)^{(1,n)}$.
\end{proof}

\begin{restatable}[Bounding: Disj, Props, Fair]{lem}{DisjBoundingLemmaFair}
\label{disj:le:FairDisjunctiveBounding}
For disjunctive systems:
\begin{align*}
&\forall n>2|B|: \\
&(A,B)^{(1,2|B|)} \models \pexists_{uncond} h(A,B_1) &
&\impliedby& &
(A,B)^{(1,n)} \models \pexists_{uncond} h(A,B_1),\\
\end{align*}
\end{restatable}
The proof was given in the main text, in Section~\ref{sec:ideas-disj-fair}.


\begin{restatable}[Disj, Props, Fair]{tightness}{TightDisjBoundingLemmaFair}
\label{obs:disj:fair_tight_prop}
The cutoff in Lemma~\ref{disj:le:FairDisjunctiveBounding} is tight, i.e., 
for any $k$ there exist process templates $(A,B)$ with $|B| = k$ 
and $\LTLmX$ formula $h(A,B_1)$ such that:
$$
(A,B)^{(1,2|B|)} \models \pexists h(A,B_1) ~~and~~ 
(A,B)^{(1,2|B|-1)} \not\models \pexists h(A,B_1).
$$
\ak{what happens if we bound $T_A$?}
\end{restatable}
The proof was described in the main text, in Section~\ref{sec:ideas-disj-fair}.
\begin{proof}
\input{tight-disj-bounding-fair}
\end{proof}

\ifwithextensions
\begin{restatable}[Generalized Bounding Lemma, Unconditionally Fair]{lem}{DisjBoundingLemmaFairGeneral}
\label{disj:le:FairDisjunctiveBounding:general}
For disjunctive systems:
\begin{align*}
&(A,B)^{(1,> 2\card{B}+n-1)} \models \pexists_{uncond} h(A,B^{(n)}) &
&~\iff~& &
(A,B)^{(1,2\card{B}+n-1)} \models \pexists_{uncond} h(A,B^{(n)}),\\
\end{align*}
\end{restatable}
\iffinal
\input{proof-disj-bounding-fair-general}
\else
\fi

\begin{restatable}{obs}{TightDisjBoundingLemmaFairGeneral}
\label{obs:disj:fair_tight_prop:general}
The cutoff in Lemma~\ref{disj:le:FairDisjunctiveBounding:general} is tight.
\end{restatable}

\input{tight-disj-bounding-fair-general}
\fi 

\begin{restatable}[Monotonicity: Disj, Deadlocks, Fair]{lem}{mono_lem_disj_deadlocks_fair}
\label{mono_lem_disj_deadlocks_fair}
For disjunctive systems, on strong-fair or finite runs:
$$
\forall n\geq |B|+1: (A,B)^{(1,n)} \textit{ has a deadlock} 
\ \Impl\ 
(A,B)^{(1,n+1)} \textit{ has a deadlock}
$$
\end{restatable}
\begin{proof}
See proof of Lemma~\ref{mono_lem_disj_deadlocks_unfair}.
\end{proof}

\begin{restatable}[Bounding: Disj, Deadlocks, Fair]{lem}{DisjDeadlockLemmaFair}
\label{le:disj:fair_tight_deadlock}
For disjunctive systems, on strong-fair or finite runs:
\li
  \- with $c=2|B|-1$ and any $n>c$:
  $$(A,B)^{(1,c)} \textit{ has a local deadlock} \ \Implied\ (A,B)^{(1,n)} \textit{ has a local deadlock}$$
  
  \- with $c=2|B| - 1$ and any $n>c$
  $$(A,B)^{(1,c)} \textit{ has a global deadlock} \ \Implied\ (A,B)^{(1,n)} \textit{ has a global deadlock} $$
  
  \- with $c=2|B|-1$ and any $n>c$:
  $$(A,B)^{(1,c)} \textit{ has a deadlock} \ \Implied\ (A,B)^{(1,n)} \textit{ has a deadlock}$$
\il
\end{restatable}
\begin{proof}
\input{proof-disj-deadlock-fair}
\end{proof}

\begin{restatable}[Disj, Deadlocks, Fair]{tightness}{TightDisjDeadlockLemmaFair}
\label{obs:disj:fair_tight_deadlock}
The cutoff $c=2|B|-1$ for deadlock detection in disjunctive systems on strong-fair or finite runs is tight, i.e.: for any $k$ there are templates $(A,B)$ with $|B|=k$ such that:
$$
(A,B)^{(1,2|B|-2)} \textit{ does not have a deadlock, but } (A,B)^{(1,2|B|-1)} \textit { does}.
$$
\end{restatable}

\input{tight-disj-deadlock-fair}

\ifwithextensions   
\subsection{Label-dependent Cutoffs for Disjunctive Systems}
\sj{To Do: compare to synchronization skeletons}

The results presented in Lemmas~\ref{disj:le:NonFairDisjunctiveBounding} and \ref{disj:le:FairDisjunctiveBounding} give cutoffs that depend only on the size of templates, without considering the number of output variables that can be used to distinguish states in transition guards. In the following, we analyze the effect of a limited number of output variables on cutoffs. Let $L_B$ be the set of state labels of process template $B$, and $\labelings_B = \bbB^{L_B}$ the set of labelings of $B$.

\begin{restatable}[Bounding Lemma, Label-dependent, Unfair]{lem}{DisjBoundingLemmaLabel}
\label{disj:le:NonFairDisjunctiveBoundingLabels}
For disjunctive systems:
\begin{align*}
&(A,B)^{(1,> \card{\labelings_B} + 2)} \models \pexists h(A,B_1) &
&~\iff~& &
(A,B)^{(1,\card{\labelings_B} + 2)} \models \pexists h(A,B_1).
\end{align*}
\end{restatable}
\iffinal
\else
\input{proof-disj-bounding-label}
\fi

\begin{restatable}[Deadlock Detection, Label-dependent, Unfair]{lem}{DisjDeadlockLemmaLabel}
\label{le:disj:deadlocks:label}
If the open disjunctive system $(A,B)^{(1,2\card{\labelings_B})}$ has no local deadlocks, then no disjunctive system $(A,B)^{(1,n)}$ with $n \geq 2\card{\labelings_B}$ has local deadlocks.
\end{restatable}
\iffinal
\else
\input{proof-disj-deadlock-label}
\fi

\begin{restatable}[Bounding Lemma, Unconditional Fairness, Label-dependent]{lem}{DisjBoundingLemmaLabelFair}
\label{disj:le:FairDisjunctiveBoundingLabels}
For unconditionally-fair runs of disjunctive systems:
\begin{align*}
&(A,B)^{(1,> 3 \card{\labelings_B} - 1)} \models \pexists h(A,B_1) &
&~\iff~& &
(A,B)^{(1,3 \card{\labelings_B} - 1)} \models \pexists h(A,B_1).
\end{align*}
\end{restatable}
\iffinal
\else
\input{proof-disj-bounding-label-fair}
\fi

\begin{restatable}[Deadlock Detection, Strong Fairness, Label-dependent]{lem}{DisjDeadlockLemmaLabelFair}
\sj{wait for clarification in fair state-based case}
\end{restatable}
\fi 

%% file: proof-disj-bounding-new.tex
\begin{proof}
Let $c = |B|+2$ and $n \geq c$. Let $x=(s_1,e_1,p_1), (s_2,e_2,p_2) \ldots$ be a run of $\largesys$ that satisfies $\pexists h(A,B_1)$. We construct a run $y$ of the cutoff system $\cutoffsys$ with $y(A, B_1) \simeq x(A, B_1)$.

Let $\visited(x)$ be the set of all visited states by B-processes in run $x$: $\visited(x) = \{ q \| \exists m \exists i: s_m(B_i) = q \}$. 

Construct the run $y$ of \cutoffsys as follows:
\li
  \-[a.] copy runs of $A$ and $B_1$ from $x$ to $y$: $y(A)=x(A)$, $y(B_1)=x(B_1)$
  \-[b.] $x$ is infinite, hence it has at least one infinitely moving process, denoted $B_\infty$. Devote one unique process $B_\infty$ in \cutoffsys that copies the behaviour of $B_\infty$ of \largesys: $y(B_\infty)=x(B_\infty)$.
  \-[c.] for every $q \in \visited$ there is a process of \largesys, denoted $B_i$, that visits $q$ first, at moment denoted $m_q$. Then devote one unique process in \cutoffsys, denoted $B_{i_q}$, that \emph{floods $q$}: set $y(B_{i_q}) = x(B_i)\slice{1}{m_q}(q)^\omega$. In words: the run $y(B_{i_q})$ repeats exactly that of $x(B_i)$ till moment $m_q$, after which the process is never scheduled.
  \-[d.] let any other process $B_i$ of \cutoffsys not used in the previous steps (if any) \emph{mimic} the behavior of $B_1$ of \cutoffsys: $y(B_i) = y(B_1)$.
\il
The figure illustrates the construction.\ak{\init should be flooded}
\begin{figure}
\centering
\scalebox{0.7}{
\input{img/disj_flooding_construction}
}
\end{figure}
The correctness follows from the observation that any transition of any process at any moment $m$ of $y$ was done by some process in $x$ at moment $m$ and hence is enabled. Also note that if $\geq 2$ processes transit simultaneously in $y$, then the guards of their transitions will be enabled even if both of them are removed from the state space\ak{vague}. Note that it is possible that in $y$:
\li
  \- more than one process transits at the same moment. Then, \emph{\interleave} the transitions of such processes, namely arbitrarily sequentialize them. \ak{why are enabled}
  \- at some moment no processes move. Then remove elements of the run $y$ -- the resulting run is denoted $\destutter(y)$.
\il
This construction uses $|\visited| + 2 \leq |B|+2$ copies of B (ignoring case (d)).
\end{proof} 

%% file: img/disj_flooding_construction.tex
\usetikzlibrary{calc,arrows,shapes,fit}

\begin{tikzpicture}
\tikzstyle{every label} = [font=\normalsize];
\tikzstyle{state} = [circle, inner sep=2pt, minimum size=2mm, draw=black, node distance=0.5cm]
\tikzstyle{hiddenstate} = [state, circle,inner sep=0pt,outer sep=0pt, minimum size=0pt, draw=gray]
\tikzstyle{topstate} = [hiddenstate, node distance=0.75cm]
\tikzstyle{state21} = [state, rectangle]
\tikzstyle{state22} = [state, diamond]
\tikzstyle{state23} = [state, regular polygon,circle]


\node (t11_top) [topstate, label={[name=t11_top-label] $A_1$}] {};
\node (t11_1) [hiddenstate, below of=t11_top] {};
\node (t11_2) [hiddenstate, below of=t11_1] {};
\node (t11_3) [hiddenstate, below of=t11_2] {};
\node (t11_4) [hiddenstate, below of=t11_3] {};
\node (t11_5) [hiddenstate, below of=t11_4] {};
\node (t11_6) [hiddenstate, below of=t11_5] {};
\draw[|-] (t11_top) -- (t11_6);

\node (t21_top) [topstate, right of=t11_top, label={[name=t21_top-label] $B_1$}] {};
\node (t21_1) [hiddenstate, below of=t21_top] {};
\node (t21_2) [hiddenstate, below of=t21_1] {};
\node (t21_3) [hiddenstate, below of=t21_2] {};
\node (t21_4) [hiddenstate, below of=t21_3] {};
\node (t21_5) [hiddenstate, below of=t21_4] {};
\node (t21_6) [hiddenstate, below of=t21_5] {};
\draw[|-] (t21_top) -- (t21_6);

\node (t22_top) [topstate, right of=t21_top, label={[name=t22_top-label] $B_2$}] {};
\node (t22_1) [state21, below of=t22_top] {};
\node (t22_2) [hiddenstate, below of=t22_1] {};
\node (t22_3) [state22, below of=t22_2] {};
\node (t22_4) [hiddenstate, below of=t22_3] {};
\node (t22_5) [state23, below of=t22_4] {};
\node (t22_6) [hiddenstate, below of=t22_5] {};
\draw[|-] (t22_top) -- (t22_1);
\draw[] (t22_1) -- (t22_3);
\draw[] (t22_3) -- (t22_5);
\draw[] (t22_5) -- (t22_6);

\node (t2i_top) [topstate, right of=t22_top] {};

\node (t2m_top) [topstate, right of=t2i_top, label={[name=t2m_top-label] $B_{n-1}$}] {};
\node (t2m_1) [hiddenstate, below of=t2m_top] {};
\node (t2m_2) [state23, below of=t2m_1] {};
\node (t2m_3) [hiddenstate, below of=t2m_2] {};
\node (t2m_4) [state22, below of=t2m_3] {};
\node (t2m_5) [hiddenstate, below of=t2m_4] {};
\node (t2m_6) [hiddenstate, below of=t2m_5] {};
\draw[|-] (t2m_top) -- (t2m_2);
\draw[] (t2m_2) -- (t2m_4);
\draw[] (t2m_4) -- (t2m_6);

\node (t2n_top) [topstate, right of=t2m_top, label={[name=t2n_top-label] $B_n$}] {};
\node (t2n_1) [hiddenstate, below of=t2n_top] {};
\node (t2n_2) [hiddenstate, below of=t2n_1] {};
\node (t2n_3) [hiddenstate, below of=t2n_2] {};
\node (t2n_4) [hiddenstate, below of=t2n_3] {};
\node (t2n_5) [hiddenstate, below of=t2n_4] {};
\node (t2n_6) [hiddenstate, below of=t2n_5, label={[name=t2n_6-label]below:$\infty$}] {};
\draw[|-] (t2n_top) -- (t2n_6);


\coordinate (t11_old-top) at (t11_top);
\coordinate (t11_old-bottom) at (t11_6);
\coordinate (t21_old-top) at (t21_top);
\coordinate (t21_old-bottom) at (t21_6);
\coordinate(t22_old-bottom) at (t22_6);
\coordinate (t2n_old-top) at (t2n_top);
\coordinate(t2m_old-bottom) at (t2m_6);
\coordinate(t2n_old-bottom) at ($(t2n_6-label)-(0,0.1)$);

\coordinate (t22_center) at ($(t22_3.center)$);
\coordinate (t2m_center) at ($(t2m_3.center)$);
\node (othrerssss) at ($(t22_center)+(0.8,0)$) {\textbf{\ldots}};

\node (box_outer) [draw=black, fit=(t11_top-label.center) (t2n_6), inner sep=0.5cm] {} ;


\node (t11_top) [topstate, label={[name=t11_top-label] $A_1$}, right of=t2n_old-top, node distance=1.5cm] {};
\node (t11_1) [hiddenstate, below of=t11_top] {};
\node (t11_2) [hiddenstate, below of=t11_1] {};
\node (t11_3) [hiddenstate, below of=t11_2] {};
\node (t11_4) [hiddenstate, below of=t11_3] {};
\node (t11_5) [hiddenstate, below of=t11_4] {};
\node (t11_6) [hiddenstate, below of=t11_5] {};
\draw[|-] (t11_top) -- (t11_6);

\node (t21_top) [topstate, right of=t11_top, label={[name=t21_top-label] $B_1$}] {};
\node (t21_1) [hiddenstate, below of=t21_top] {};
\node (t21_2) [hiddenstate, below of=t21_1] {};
\node (t21_3) [hiddenstate, below of=t21_2] {};
\node (t21_4) [hiddenstate, below of=t21_3] {};
\node (t21_5) [hiddenstate, below of=t21_4] {};
\node (t21_6) [hiddenstate, below of=t21_5] {};
\draw[|-] (t21_top) -- (t21_6);

\node (t22_top) [topstate, right of=t21_top, label={[name=t22_top-label] $B_2$}] {};
\node (t22_1) [state21, below of=t22_top] {};
\node (t22_2) [state21, below of=t22_1] {};
\node (t22_3) [state21, below of=t22_2] {};
\node (t22_4) [state21, below of=t22_3] {};
\node (t22_5) [state21, below of=t22_4] {};
\node (t22_6) [hiddenstate, below of=t22_5] {};
\draw[|-] (t22_top) -- (t22_1);
\draw[dashed] (t22_1) -- (t22_2);
\draw[dashed] (t22_2) -- (t22_3);
\draw[dashed] (t22_3) -- (t22_4);
\draw[dashed] (t22_4) -- (t22_5);
\draw[dashed] (t22_5) -- (t22_6);

\node (t2i_top) [topstate, right of=t22_top] {};

\node (t23_top) [topstate, right of=t22_top, label={[name=t23_top-label] $B_3$}] {};
\node (t23_1) [state21, below of=t23_top] {};
\node (t23_2) [hiddenstate, below of=t23_1] {};
\node (t23_3) [state22, below of=t23_2] {};
\node (t23_4) [state22, below of=t23_3] {};
\node (t23_5) [state22, below of=t23_4] {};
\node (t23_6) [hiddenstate, below of=t23_5] {};
\draw[|-] (t23_top) -- (t23_1);
\draw[] (t23_1) -- (t23_3);
\draw[dashed] (t23_3) -- (t23_4);
\draw[dashed] (t23_4) -- (t23_5);
\draw[dashed] (t23_5) -- (t23_6);

\node (t24_top) [topstate, right of=t23_top, label={[name=t24_top-label] $B_4$}] {};
\node (t24_1) [hiddenstate, below of=t24_top] {};
\node (t24_2) [state23, below of=t24_1] {};
\node (t24_3) [state23, below of=t24_2] {};
\node (t24_4) [state23, below of=t24_3] {};
\node (t24_5) [state23, below of=t24_4] {};
\node (t24_6) [hiddenstate, below of=t24_5] {};
\draw[|-] (t24_top) -- (t24_2);
\draw[dashed] (t24_2) -- (t24_3);
\draw[dashed] (t24_3) -- (t24_4);
\draw[dashed] (t24_4) -- (t24_5);
\draw[dashed] (t24_5) -- (t24_6);

\node (t25_top) [topstate, right of=t24_top, label={[name=t25_top-label] $B_5$}] {};
\node (t25_1) [hiddenstate, below of=t25_top] {};
\node (t25_2) [hiddenstate, below of=t25_1] {};
\node (t25_3) [hiddenstate, below of=t25_2] {};
\node (t25_4) [hiddenstate, below of=t25_3] {};
\node (t25_5) [hiddenstate, below of=t25_4] {};
\node (t25_6) [hiddenstate, below of=t25_5, label={[name=t25_6-label]below:$\infty$}] {};
\draw[|-] (t25_top) -- (t25_6);

\coordinate (t11_new-top) at (t11_top);
\coordinate (t11_new-bottom) at (t11_6);
\coordinate (t21_new-top) at (t21_top);
\coordinate (t21_new-bottom) at (t21_6);
\coordinate (t22_new-top) at (t22_top);
\coordinate (t22_new-bottom) at (t22_6);
\coordinate(t23_new-bottom) at (t23_6);
\coordinate(t24_new-bottom) at (t24_6);
\coordinate(t25_new-bottom) at ($(t25_6-label)-(0,0.1)$);

\node (box_outer) [draw=black, fit=(t11_top-label.center) (t25_6), inner sep=0.5cm] {} ;

\draw[->, bend right, dotted,semithick] (t22_old-bottom) to (t22_new-bottom);
\draw[->, bend right, dotted,semithick] (t22_old-bottom) to (t23_new-bottom);
\draw[->, bend right, dotted,semithick] (t2m_old-bottom) to (t24_new-bottom);
\draw[->, bend left=5, dotted,semithick] (t2n_top-label.north) to (t25_top-label.north);

\end{tikzpicture}

%% file: tight-disj-bounding.tex
\begin{proof}
The idea of the proof relies on the subtleties of the definition of a run: it is infinite (thus not globally deadlocked), and in each step of a run exactly one process moves. 

Consider the templates in the figure below and let $\pexists h(A,B_1) = \pexists (\eventually 3_{B_1} \land \eventually\always (2_{B_1} \land {end}_A))$. In words: there exists a run in a system where process $B_1$ visits $3_B$ and process $B_1$ with $A$ eventually always stay in $2_B$ and ${end}_A$.
\begin{figure}[h]
\centering
\subfloat[Template A]{
\centering
\makebox[0.4\textwidth][c]{
\scalebox{0.75}{\input{img/disj_tight_propAB_tmplA}}
\label{fig:disj:tight_propAB_tmplA}
}}
\subfloat[Template B]{
\centering
\makebox[0.6\textwidth][c]{
\scalebox{0.75}{\input{img/disj_tight_propAB_tmplB}}
\label{fig:disj:tight_propAB_tmplB}
}}
\label{fig:disj:tight_propAB_tmpl}
\end{figure}

We need one process in every state of $B$ to enable the transitions of $A$ to ${all}_A$. Only when $A$ in ${all}_A$, $B_1$ can move $3_B \to 1_B$, and then at some point to $2_B$. After $B_1$ moves $3_B \to 1_B$, $A$ moves ${all}_A \to {end}_A$ which requires process $B_{i \neq 1}$ in $3_B$. Finally, to make the run infinite there should be at least two processes in ${|B|}_B$.\sj{other cases are covered in general lemma below}
\end{proof}

%% file: img/disj_tight_propAB_tmplA.tex
\begin{tikzpicture}[node distance=1.8cm,inner sep=1pt,minimum size=0.5mm,->,>=latex]

\node[initial below, state] (a_1) {$1_A$};
\node (dots) [right of=a_1] {$\ldots$};
\node[state] (a_all) [right of=dots] {${all}_A$};
\node[state] (a_end) [right of=a_all] {${end}_A$};

\path (a_1) edge [above] node {$\disj{1_B}$} (dots);
\path (dots) edge [above] node {$\disj{{\card{B}}_B}$} (a_all);
\path (a_all) edge [above] node {$\disj{3_B}$} (a_end);

\end{tikzpicture}

%% file: img/disj_tight_propAB_tmplB.tex
\begin{tikzpicture}[node distance=1.8cm,inner sep=1pt,minimum size=0.5mm,->,>=latex]

\node[initial below, state] (b_1) {$1_B$};
\node[state] (b_2) [left of=b_1] {$2_B$};
\node[state] (b_3) [right of=b_1]{$3_B$};
\node (dots) [right of=b_3] {$\ldots$};
\node[state] (b_k) [right of=dots] {${\card{B}}_B$}; 

\path (b_1) edge [above] node {$\disj{1_B}$} (b_2);
\path (b_1) edge [above] node {$\disj{1_B}$} (b_3);
\path (b_3) edge [above] node {$\disj{3_B}$} (dots);
\path (dots) edge [above] node {$\disj{{\card{B}{-}1}_B}$} (b_k);
\path (b_k) [loop above] edge [right] node {$\disj{{\card{B}}_B}$} (b_k);
\path (b_3) [bend right=60] edge [above] node {$\disj{{all}_A}$} (b_1);

\end{tikzpicture}

%% file: proof-disj-bounding-general.tex
\begin{proof}
An inspection of the proof of Lemma~\ref{disj:le:NonFairDisjunctiveBounding} shows that it works almost without modification if we replicate the local run of $A^1$ along with $n$ local runs of copies of $B$, instead of only $A^1, B^1$. In particular, the rest of the flooding construction and the interleaving construction do not have to be changed.

Additionally, one can observe that the cutoff is independent of whether the property talks about $A^1$ or not.
\end{proof}

%% file: tight-disj-bounding-general.tex
\begin{proof}
We can again use the templates from the previous observation, with a formula $h(A^1,B^{(n)}) = \pexists \bigwedge_{i \in [1..n]} (\eventually b^i_3 \land \eventually\always (b^i_2 \land a^1_{end}))$.\footnote{The case of $n=0$ can be shown with a simpler templates: $A$ is the template with only one state $a_{end}$ without successors, and $B$ is the chain that ends in $b_k$ with the self guarded loop.}
\end{proof}

%% file: proof-disj-deadlock.tex
Given a (globally or locally) deadlocked run of $\largesys$ 
we construct (globally or locally) deadlocked run of $\cutoffsys$, 
where $c$ depends on the nature of the given run. 
We do this using the construction template. 

Let $\mB=\{B_1,...,B_n\}$.
The template depends on set $\mC \subseteq \{B_1,...,B_c\}$:
\li
  \-[a.] set $y(A)=x(A)$
  \-[b.] for every $B_i \in \mC$, set $y(B_i)=x(B_i)$
  \-[c.] for every $q \in \visInf{\mB\smi\mC}{x}$, 
         devote one process of \cutoffsys that floods $q$
  \-[d.] for every $q \in \visFin{\mB\smi\mC}{x}$, 
         devote one process of \cutoffsys that floods $q$ 
         and then evacuates into $\visInf{\mB\smi\mC}{x}$
  \-[e.] let other processes (if any) mimic some process from (c)
\il

\myparagraph{1) Local Deadlock}
We distinguish three cases: 
\li
  \-[1a)] $A$ deadlocks, $B_1$ moves infinitely often
  \-[1b)] $A$ moves infinitely often, $B_1$ deadlocks
  \-[1c)] $A$ neither deadlocks nor moves infinitely often, 
          $B_1$ deadlocks, $B_2$ moves infinitely often.
\il

\myparagraphraw{1a:} ``$A$ deadlocks, $B_1$ moves infinitely often''. 

Let $c=|B|+1$, and $\mC=\{B_1\}$.
Note that $\visInf{B_2..B_n}{x} \neq \emptyset$. 
The resulting construction uses 
$|\visFin{B_2..B_n}{x}| + |\visInf{B_2..B_n}{x}| + 1 
 \leq 
 |B| + 1$ 
copies of B.
\ak{seems tight}\ak{correctness}

\myparagraphraw{1b:} ``$A$ moves infinitely often, $B_1$ deadlocks''. 

Let $c=|B|+1$, and $\mC=\{B_1\}$.
Let $q_\bot$ be the state in which $B_1$ deadlocks.
Instantiate the construction template.

Process $B_1$ of \cutoffsys is deadlocked in $y$ starting from some moment $d$,
because any state it sees (in $\visInf{A,B_2..B_n}{x}$)
was also seen by $B_1$ in \largesys in $x$ at some moment $d' \geq d$
(note that $d'$ may be not the same moment as $d$).

\myparagraphraw{1c:} ``$A$ neither deadlocks nor moves infinitely often, 
                       $B_1$ deadlocks, $B_2$ moves infinitely often''. 

Instantiate the construction template with $c=|B|+2$ and $\mC = \{B_1,B_2\}$.
\ak{seems not tight}\ak{correctness}

\smallskip
Finally, $|B|+2$ is a (possibly not tight) cutoff for local deadlock detection problem.

\myparagraph{2) Global Deadlock}
Let $x=(s_1,e_1,p_1)...(s_d,e_d,\bot)$ be a globally deadlocked run of $\largesys$ 
with $n\geq c$.

Let us abuse the definition of $\visInf{\mF}{x}$ and $\visFin{\mF}{x}$,
in Eq.~\ref{disj:def_vinf_wrt} and \ref{disj:def_vfin_wrt} resp., 
and adapt it to the case of finite runs.
To this end, given a finite run $x=(s_1,e_1,p_1)...(s_d,e_d,\bot)$, 
extend it to the infinite sequence $(s_1,e_1,p_1)...(s_d,e_d,\bot)^\omega$, 
and apply the definition of $\visInf{\mF}{x}$ and $\visFin{\mF}{x}$ to the sequence.

Let $\mD_1$ be the set of processes deadlocked in unique states:
$\forall p\in \mD_1 \not\exists p' \neq p: s_d(p')=s_d(p)$.
Instantiate the construction template with $\mC = \mD_1$ and $c=2|B|-1$.
\footnote{$2|B|-1$ copies is enough, because: 
          $\visFin{\mB\smi\mC}{x} \cap \visInf{\mB\smi\mC}{x} = \emptyset$,
          $\visInf{\mB\smi\mC}{x} \cap \visInf{\mC}{x} = \emptyset$,
          and if $\visFin{\mB\smi\mC}{x} \neq \emptyset$, 
          then $\visInf{\mB\smi\mC}{x} \neq \emptyset$.}
\ak{seems not tight}


\myparagraph{3) Deadlocks}
As the cutoff for the deadlock detection problem we take the largest cutoff in (1)-(2), namely, $2|B|-1$, but it may be not tight -- finding the tight cutoffs for local deadlock and for deadlock detection problems is an open problem.

\ak{tried to refine but could not -- the trial is commented out}

\sj{my idea for smaller cutoff in comments}

%
%
%
%

%% file: tight-disj-deadlock.tex
\begin{proof}
The figure below illustrates templates $(A,B)$ to prove the asymptotical optimality of cutoff $2|B|-1$ for deadlock detection problem. Template $A$ is any that never deadlocks. The system has a local deadlock only when there are at least $|B|$ copies of $B$, which is a constant factor of $2|B|-1$.
\begin{figure}[h]
\centering
\makebox[0.4\textwidth][c]{
\scalebox{0.75}{\input{img/disj_tight_deadlock_tmpl}}
}
\end{figure}
%
\end{proof}


%% file: img/disj_tight_deadlock_tmpl.tex
\begin{tikzpicture}[node distance=1.8cm,inner sep=1pt,minimum size=0.5mm,->,>=latex]

\node[initial left, state] (b_1) {$1_B$};
\node[state] (b_2) [right of=b_1] {$2_B$};
\node (dots) [right of=b_2] {$\ldots$};
\node[state] (b_k) [right of=dots] {$k_B$}; 

\path (b_1) edge [above] node {$\disj{1_B}$} (b_2);
\path (b_2) edge [above] node {$\disj{2_B}$} (dots);
\path (dots) edge [above] node {$\disj{{k-1}_B}$} (b_k);

\path (b_1) [loop above] edge [right] node {} (b_1);
\path (b_2) [loop above] edge [right] node {} (b_2);

\path (b_1) edge [below, bend right=22] node {} (b_k);
\path (b_2) edge [below, bend right=15] node [near start] {$\disj{{k}_B}$} (b_k);

\end{tikzpicture}

%
%
%
%

%% file: tight-disj-bounding-fair.tex
Consider process templates $A,B$ in the figure below 
and property $\pexists \true$.
\begin{figure}[h]
\centering
\subfloat[Template A]{
\centering
\scalebox{0.75}{\input{img/disj_fair_tight_propAB_tmplA}}
\label{fig:disj:fair_tight_propAB_tmplA}
}
\hspace{1cm}
\subfloat[Template B]{
\centering
\scalebox{0.75}{\input{img/disj_fair_tight_propAB_tmplB}}
\label{fig:disj:fair_tight_propAB_tmplB}
}
\end{figure}

%% file: img/disj_fair_tight_propAB_tmplA.tex
\begin{tikzpicture}[node distance=1.8cm,inner sep=1pt,minimum size=0.5mm,->,>=latex]

\node[state,initial left] (a_1) {$1_A$};
\node (dots) [right of=a_1] {$\ldots$};
\node[state] (a_all) [right of=dots] {${all}_A$};

\path (a_1) edge [above] node {$\disj{1_B}$} (dots);
\path (dots) edge [above] node {$\disj{k_B}$} (a_all);
\path (a_all) [loop above] edge node {} (a_1);

\end{tikzpicture}

%% file: img/disj_fair_tight_propAB_tmplB.tex
\begin{tikzpicture}[node distance=1.8cm,inner sep=1pt,minimum size=0.5mm,->,>=latex]

\node[state,initial left] (b_1) {$1_B$};
\node[state] (b_2) [right of=b_1] {$2_B$};
\node (dots) [right of=b_2] {$\ldots$};
\node[state] (b_k) [right of=dots] {$k_B$}; 

\path (b_1) edge [above] node {$\disj{1_B}$} (b_2);
\path (b_2) edge [above] node {$\disj{2_B}$} (dots);
\path (dots) edge [above] node {$\disj{{k-1}_B}$} (b_k);

\path (b_1) [loop above] edge [right] node {$\disj{1_B}$} (b_1);
\path (b_2) [loop above] edge [right] node {$\disj{2_B}$} (b_2);
\path (b_k) [loop above] edge [right] node {$\disj{k_B}$} (b_k);

\end{tikzpicture}

%% file: proof-disj-bounding-fair-general.tex
\begin{proof}
Again, an inspection of the proof of Lemma~\ref{disj:le:FairDisjunctiveBounding} shows that it works almost without modification if we replicate the local run of $A^1$ along with $k$ local runs of copies of $B$, instead of only $A^1, B^1$. In particular, the rest of the flooding construction and the interleaving construction do not have to be changed.
\end{proof}

%% file: tight-disj-bounding-fair-general.tex
\begin{proof}
Consider again process templates $A,B$ in Figure~\ref{fig:disj:fair_tight_propAB_tmpl}, with property $\pexists \bigwedge_{i \in [1..n]} \GF b_1$, i.e., in addition to enabling an infinite run, we need to keep $n$ processes in $b_1$.
\end{proof}

%% file: proof-disj-deadlock-fair.tex
\providecommand{\deadOne}[1]{\dead_{<2}(#1)}
\providecommand{\deadTwo}[1]{\dead_2(#1)}

If $\largesys$ has a global deadlock, 
then the fairness does not influence the cutoff, 
and the proof from Lemma~\ref{lem_disj_deadlocks_unfair}, 
case ``Global Deadlocks'', applies giving the cutoff $2|B|-1$. 
Hence below consider only the case of local deadlocks. 

Given a strong-fair deadlocked run $x$ of $\largesys$, 
we first construct a strong-fair deadlocked run $y$ of $\cutoffsys$ 
with $c=2|B|$ and then argue that $c$ can be reduced to $2|B|-1$. 
The construction is similar to that in Lemma~\ref{lem_disj_deadlocks_unfair} 
-- the differences originate from the need to infinitely move 
non deadlocked processes.

Let $\deadOne{x}$ be the set of deadlocked states in the run $x$ 
that are only deadlocked if there is no other process in the same state, 
and let $\mD_1$ be the set of processes deadlocked 
in the run $x$ in $\deadOne{x}$. 
Let $\deadTwo{x}$ be the set of states that are deadlocked 
in the run $x$ even if there is another process in the same state. 

Notes: 
\li
\- $|\mD_1| = |\deadOne{x}| \leq |B|$

\- $\deadOne{x} \cap \deadTwo{x} = \emptyset$

\- $\visFin{\mB\smi \mD_1}{x} \cap \deadOne{x} \neq \emptyset$
   is possible, because a state from $\visFin{\mB\smi \mD_1}{x}$ 
   can first be visited by a process in $\mB \smi \mD_1$, 
   and later deadlocked by the process in $\mD_1$.

\- $\deadTwo{x} \subseteq \visInf{\mB\smi \mD^1}{x}$,
   and hence $\visFin{\mB \smi \mD_1}{x} \cap \deadTwo{x} = \emptyset$.
\il

The construction has two phases. 
The first phase:
\li
\-[a.] for every $p \in \{A\} \cup \mD_1$, set $y(p)=x(p)$ 

\-[b.] for every $q \in \deadTwo{x}$, 
       devote one process of $\cutoffsys$ that floods it

\-[c.] for every $q \in \visInf{\mB \smi \mD_1}{x} \smi \deadTwo{x}$,
       devote two processes of $\cutoffsys$ that flood it

\-[d.] for every $q \in \visFin{\mB \smi \mD_1}{x}$, 
       devote one process of $\cutoffsys$ that floods it
       and then evacuates into $\visInf{\mB \smi \mD_1}{x}$

\-[e.] let other processes (if any) mimic some process from (c)
\il
After this phase all $B$ processes will be in 
$\visInf{\mB \smi \mD_1}{x} \cup \deadOne{x}$. 

The second phase applies to processes in 
$\visInf{\mB \smi \mD_1}{x} \smi \deadTwo{x}$ the fair extension%
\footnote{The fair extension requires run $x$ to be unconditionally-fair, 
          but here we have a run in which all processes that are not deadlocked
          move infinitely often.
          To adapt the construction to this case:
          copy local runs of processes $\{A\} \cup \mD_1$,
          and do not extend local runs of processes that are in 
          state in $\dead_2$.}.


How many processes does the construction use? 
Note that the sets 
$\deadOne{x} \cup \visFin{\mB \smi \mD_1}{x}$, 
$\deadTwo{x}$, 
$\visInf{\mB \smi \mD_1}{x} \smi \deadTwo{x}$ 
are disjoint, thus:
\begin{align}
& | \visFin{\mB \smi \mD_1}{x} | + |\deadOne{x}| + |\deadTwo{x}| + 2| \visInf{\mB \smi \mD_1}{x} \smi \deadTwo{x} | \leq \label{disj:eq:1} \\
& 2|\visFin{\mB \smi \mD_1}{x} \cup \deadOne{x}| + |\deadTwo{x}| + 2| \visInf{\mB \smi \mD_1}{x} \smi \deadTwo{x} | \leq \label{disj:eq:2} \\
& |B| + |\visFin{\mB \smi \mD_1}{x} \cup \deadOne{x}| + | \visInf{\mB \smi \mD_1}{x} \smi \deadTwo{x} | \leq 2|B|  \nonumber
\end{align}
Let us reduce the estimate to $\leq 2|B|-1$:
\li
  \- assume that $\deadTwo{x} = \emptyset$ 
     (otherwise, Eq.\ref{disj:eq:1} and the sets disjointness give $2|B|-1$)

  \- assume that $ \visFin{\mB \smi \mD_1}{x} \neq \emptyset$ 
     (the other case together with eq.\ref{disj:eq:2}, 
      the sets disjointness, and the first item gives $2|B|-1$)

  \- hence, the construction in step (d) evacuates the process in 
     $q \in \visFin{\mB \smi \mD_1}{x}$ 
     into 
     $ \visInf{\mB \smi \mD_1}{x} \smi \deadTwo{x}$. 
     Hence modify step (c) of the construction 
     and for $q$ devote a single process of $\cutoffsys$ that floods it. 
     This will give $\leq 2|B|-1$.
\il
This concludes the proof.

%% file: tight-disj-deadlock-fair.tex
\begin{proof}
The figure below shows process templates $(A,B)$ such that any system $\largesys$ with $n\leq 2|B|-2$ does not deadlock on strong-fair runs, but larger systems do.
\begin{figure}[h]
\centering
\vspace{-10pt}
\subfloat[Template A]{
\centering
\scalebox{0.75}{\input{img/disj_tight_fair_deadlock_tmplA}}
\label{fig:disj:tight_fair_deadlock_tmplA}
}
\hspace{1cm}
\subfloat[Template B]{
\centering
\scalebox{0.75}{\input{img/disj_tight_fair_deadlock_tmplB}}
\label{fig:disj:tight_fair_deadlock_tmplB}
}
\label{fig:disj:fair_tight_dead_tmpl}
\end{figure}
\end{proof}

%% file: img/disj_tight_fair_deadlock_tmplA.tex
\begin{tikzpicture}[node distance=1.8cm,inner sep=1pt,minimum size=0.5mm,->,>=latex]

\node[initial left, state] (a_1) {};
\node (dots) [right of=a_1] {$\ldots$};
\node[state] (a_k) [right of=dots] {}; 
\node[state] (r) [below=0.8cm of dots] {$r_A$};

\path (a_1) edge [above] node {$\disj{1_B}$} (dots);
\path (dots) edge [above] node {$\disj{{k-1}_B}$} (a_k);
\path (a_1) edge [bend right=15] node {} (r);
\path (a_k) edge [bend left=15] node {} (r);
\path (dots) edge [above, dotted] node {} (r);
\path (a_k) edge [bend right=35,above] node {$\disj{k_B}$} (a_1);

\path (r) [loop below] edge [right] node {} (r);

\end{tikzpicture}

%% file: img/disj_tight_fair_deadlock_tmplB.tex
\begin{tikzpicture}[node distance=1.8cm,inner sep=1pt,minimum size=0.5mm,->,>=latex]

\node[state,initial left] (b_1) {$1_B$};
\node (dots) [right of=b_1] {$\ldots$};
\node[state] (b_k) [right of=dots] {$k_B$}; 
\node[state] (r) [below=0.8cm of dots] {};

\path (b_1) edge [above] node {$\disj{1_B}$} (dots);
\path (dots) edge [above] node {$\disj{{k-1}_B}$} (b_k);
\path (b_1) edge [bend right=15] node {} (r);
\path (b_k) edge [bend left=15] node {} (r);
\path (dots) edge [above, dotted] node {} (r);

\path (b_1) [loop above] edge [right] node {$\disj{1_B}$} (b_1);
\path (b_k) [loop above] edge [right] node {$\disj{k_B}$} (b_k);
\path (r) [loop below] edge [right] node {$\disj{r_A}$} (r);
\path (dots) [loop above,dotted,in=60,out=120,distance=12mm] edge [right] node {} (dots);

\end{tikzpicture}

%% file: proof-disj-bounding-label.tex
\begin{proof}
The proof goes along the same lines as the proof of Lemma~\ref{disj:le:NonFairDisjunctiveBounding}. The main difference is that in the flooding construction, we only need to consider states with different labelings. That is, the flooding construction is modified, as to only include one local run for every state \emph{labeling} that is visited in the original run. Thus, the maximal number of local runs needed to ensure that the behavior of $A^1$ and $B^1$ is possible is $\card{\labelings_B}$.

The rest of the proof works just as before. In particular, we also need the local runs of $A^1$ and $B^1$, as well as possibly a copy $B^\infty$ of a local run that is infinite. Thus, the cutoff for process template $B$ is $\card{\labelings_B} + 2$.
\end{proof}

%% file: proof-disj-deadlock-label.tex
\begin{proof}
The proof works in the same way as the proof of Theorem~\ref{le:disj:deadlocks}, except:
\begin{enumerate}
\item instead of visited states we need to consider visited state labelings
\item we keep $x(A^1,B^1)$, wlog assuming that one of them is locally deadlocked in $x$
\item we use flooding for all state labelings
\item The main technical difference is that for \emph{finitely} often visited state labelings we need a modified evacuation construction, with multiple copies per labeling: 
\begin{itemize}
\item for a process $B^{\witfirst_L}$ that visits label $L$ first, and process $B^{\witlast_L}$ that visits labeling $L$ last in $x$, we cannot be sure that we can append the path from $B^{\witlast_L}$ to the path of $B^{\witfirst_L}$, since they may be in different states $q,q'$ with labeling $L$. 
\item Furthermore, another process with labeling $L'$ may be required to leave state $q$, and $L'$ may itself be a finitely visited state labeling. If $L'$ is evacuated before $B^{\witlast_L}$ reaches state $q'$ in $x$, then $B^{\witfirst_L}$ also needs to leave $q$. To ensure that we will still have a process with labeling $L$ (if it exists in $x$), we need another copy $B^{\witsecond_L}$ that is in a state $q''$ labeled with $L$ after $L'$ is evacuated (and thus, its evacuation cannot depend on labeling $L'$). If such a process does not exist in $x$, then it will also not be needed in the cutoff system. Since the evacuation from $q''$ may depend on another labeling $L''$, we may have to repeat the procedure, adding more copies as witnesses for labeling $L$.
\item To compute how many copies of $B$ are needed in total, let $\{L_1,\ldots,L_n\} = \visited_{fin}$, where $L_i$ is evacuated before $L_j$ in $x$ if $i < j$. Then for $L_1$ we only need one copy of $B$, since evacuation of $L_1$ cannot depend on any labelings that are evacuated before. In general, for labeling $L_i$ we need $i$ copies of $B$\ak{why $i$ is enough? why cannot happen smth like this -- }, since there are $i-1$ labelings that are evacuated before $L_i$, and for each of these we may need an additional copy of $B$ as a witness for $L_i$.
\end{itemize}
\item Overall, we thus need (roughly) $\card{\visited_{inf}} + \sum_{i=1}^{\card{\visited_{fin}}} i$ copies of $B$. The worst case is $\visited_{fin} = \labelings_B \setminus \{L\}$, since at least one labeling $L$ needs to be present infinitely often. Thus, the cutoff is $1 + \frac{\card{\labelings_B}\cdot(\card{\labelings_B}+1)}{2}$.
\end{enumerate}
\end{proof}

%% file: proof-disj-bounding-label-fair.tex
\begin{proof}
The proof goes along the same lines as the proof of Lemma~\ref{disj:le:FairDisjunctiveBounding}, with modified flooding and fair extension constructions. 

{\bf Original construction.} Remember that for systems with fairness, the flooding construction separates states of the originial run into those that are visited finitely often ($\visited_{fin}$) and those that are visited infinitely often ($\visited_{fin}$). Then, we consider a moment in time $\pref$, which is the first moment such that no states from $\visited_{fin}$ will be visited anymore. The original construction then ensured that the transitions of all local runs are possible until moment $\pref$, and that at this moment there are at least two processes in any given state from $\visited_{inf}$.

Starting from $\pref$, the original proof continues with the \emph{fairing extension} construction that ensures that all local runs can move infinitely often. For the combination of flooding and fair extension, it is crucial that for any state $q$ in $\visited_{inf}$, we combine a prefix from a local run that reaches $q$ first with a postfix (from a possibly different local run) that reaches $q$ infinitely often.

{\bf Modified construction.}
Like in the proofs of Lemmas~\ref{disj:le:NonFairDisjunctiveBoundingLabels} and \ref{le:disj:deadlocks:label}, we define $\visited$ and flooding based on labelings. 

For labelings in $\visited_\fin$, we use the modified evacuation construction from Lemma~\ref{le:disj:deadlocks:label}. That is, we order labelings $\{L_1,\ldots,L_n\} = \visited_\fin$ according to the moment they are evacuated, and add $i$ copies of $B$ for each $L_i$, such that we have a copy that is in a state with labeling $L_i$ before and after all labels that are evacuated before have been evacuated. All of these process copies eventually move to a loop with only labelings from $\visited_\inf$, just like they did in $x$.

For labelings $L \in \visited_\inf$, we use a combination of the original fair extension construction with this modified evacuation construction. The reason we need evacuation is the same as in Lemma~\ref{le:disj:deadlocks:label}: the state with labeling $L$ that is reached first in $x$ may be different from the state $q'$ with labeling $L$ in which some process eventually loops. Therefore, we add one copy $B^{\witfirst_L}$ that takes the shortest path to $L$, and two copies of $B^{\witlast_L}$ that go to a state $q'$ in which there is a loop coming back to $q'$ (that only depends on labelings in $\visited_\inf$). Like in Lemma~\ref{le:disj:deadlocks:label}, we may have the case that a labeling from $\visited_\fin$ needed to leave $q$ will be evacuated before $q'$ is reached by $B^{\witlast_L}$, and therefore we in general need multiple copies of $B$ that reach states with $L$ at different times. Since $L \not\in \visited_\fin$, in general all labelings from $\visited_\fin$ could be evacuated before $q'$ is reached. Thus, for every labeling $L \in \visited_\inf$, we need up to $\card{\visited_\fin}+1$ copies of $B$ that reach $L$ and eventually leave it again, and two additional copies for the fair extension that ensures that $L$ will always be present after $\pref$.

Finally, we use the interleaving construction as before.

In summary, we need $\sum_{i=1}^{\card{\visited_\fin}} i = \frac{\card{\visited_\fin} \cdot (\card{\visited_\fin}+1)}{2}$ copies of $B$ for the labelings in $\visited_\fin$, and $(\card{\visited_\fin}+3)\cdot\card{\visited_\inf}$ copies of $B$ for the labelings in $\visited_\inf$. In the worst case, this means that the cutoff for process template $B$ is (bounded by) \ldots \remove{$3 \card{\labelings_B}$. By a closer analysis, we can see that for at least one labeling $l$ in $\visited_{inf}$, the local run that visits $l$ first must also visit it infinitely often, and we only need one additional copy of this run. Thus, we obtain a cutoff of $3 \card{\labelings_B} - 1$.}
\end{proof}

%% file: appendix-conj.tex
\section{Cutoffs for Conjunctive Systems}
\label{sec:app-conj}
%

\subsection{Conjunctive Systems without Fairness}

\begin{restatable}[Monotonicity: Conj, Props, Unfair]{lem}{ConjMonotonicityLemma}
\label{le:ConjMonotonicityLemma}
For conjunctive systems,
\begin{align*}
\forall n \geq 1:\ 
(A,B)^{(1,n)} \models \pexists h(A,B_1)
\ \ \Impl \ \ 
(A,B)^{(1,n+1)}\models \pexists h(A,B_1).
\end{align*}
\end{restatable}
\begin{proof}
Let the new process stutter in $\init$ state.
\end{proof}

\begin{restatable}[Bounding: Conj, Props, Unfair]{lem}{ConjBoundingLemma}
\label{le:ConjunctiveBoundingLemma}
For conjunctive systems,
\begin{align*}
\forall n \geq 2:\ 
(A,B)^{(1,2)} \models \pexists h(A,B_1)
\ \ \Implied \ \ 
\largesys \models \pexists h(A,B_1).
\end{align*}
\end{restatable}

\begin{proof}
\input{proof-conj-bounding}
\end{proof}

\begin{restatable}[Conj, Props, Unfair]{tightness}{TightConjPropRestricted}
\label{obs:conj:tight_prop}
The cutoff $c=2$ is tight for parameterized model checking of properties $\pexists h(A,B_1)$ in the 1-conjunctive systems, i.e., there is a system type $(A,B)$ and property $Eh(A,B_1)$ which is not satisfied by $(A,B)^{(1,1)}$ but is by $(A,B)^{(1,2)}$.
\end{restatable}
\begin{proof}
The figure below shows templates $(A,B)$, $\pexists h(A,B_1) = \pexists \eventually b$. An infinite run that satisfies the formula needs one copy of $B$ that stays in the initial state, and one that moves into $b$.
\begin{figure}[h]
\centering
\subfloat[Template A]{
\centering
\scalebox{0.75}{\input{img/conj_tight_prop_tmplA}}
\label{fig:conj:tight_prop_fair_tmplA}
}
\hspace{1cm}
\subfloat[Template B]{
\centering
\scalebox{0.75}{\input{img/conj_tight_prop_tmplB}}
\label{fig:conj:tight_prop_fair_tmplB}
}
\end{figure}
\end{proof}

\ifwithextensions 
\begin{restatable}[Generalized Bounding Lemma]{lem}{ConjBoundingLemmaGeneral}
\label{le:ConjunctiveBoundingLemmaGen}
For conjunctive systems,
\begin{align*}
&\forall n \geq k+1: \largesys \models \pexists h(A,B_{(k)}) 
\ \iff\ 
(A,B)^{(1,k+1)} \models \pexists h(A,B_{(k)})
\end{align*}
\end{restatable}

\input{proof-conj-bounding-general}

\begin{restatable}[Tightness]{obs}{TightConjPropRestrictedGeneral}
\label{obs:conj:tight_prop:general}
The cutoff $k+1$ for properties $\pexists h(A,B_{(k)})$ on unconstrained runs in 1-guard conjunctive systems, i.e., there is a system type $(A,B)$ and property $Eh(A,B_{(k)})$ which is not satisfied by $(A,B)^{(1,k)}$ but is satisfied by $(A,B)^{(1,k+1)}$.
\end{restatable}
\begin{proof}
Consider again Fig.~\ref{fig:conj:tight_prop_tmpl}, and the property $\pexists h(A,B_{(k)}) = \pexists \bigwedge_{i \in [1..k]} \eventually b_i$.
\end{proof}
\fi      


%
%

\begin{restatable}[Monotonicity: Conj, Deadlocks, Unfair]{lem}{ConjDeadlockMonotonicityLemmaRestricted}
\label{le:ConjunctiveMonotonicityLemmaDeadlocks}
For conjunctive systems:
$$
\forall n\geq 1: (A,B)^{(1,n)} \textit{ has a deadlock} 
\ \Impl\ 
(A,B)^{(1,n+1)} \textit{ has a deadlock}
$$
%
%

\end{restatable}
\begin{proof}
\input{proof-monotonicity-conj-deadlock}
\end{proof}

\begin{restatable}[Bounding: 1-Conj, Deadlocks, Unfair]{lem}{ConjDeadlockLemmaRestricted}
\label{le:ConjunctiveBoundingLemmaDeadlocks}
For 1-conjunctive systems:
\li
  \- with $c=2|Q_B\smi \{ \init \}|$ and any $n>c$ \footnote{This statement also applies to systems without restriction to $1$-conjunctive guards.}
  $$(A,B)^{(1,c)} \textit{ has a global deadlock} \ \Implied\ (A,B)^{(1,n)} \textit{ has a global deadlock} $$
  
  \- with $c=|Q_B\smi \{ \init \}|+2$ and any $n>c$:
  $$(A,B)^{(1,c)} \textit{ has a local deadlock} \ \Implied\ (A,B)^{(1,n)} \textit{ has a local deadlock}$$
  
  \- with $c=2|Q_B\smi \{ \init \}|$ and any $n>c$:
  $$(A,B)^{(1,c)} \textit{ has a deadlock} \ \Implied\ (A,B)^{(1,n)} \textit{ has a deadlock}$$
\il
\end{restatable}
\begin{proof}
\input{proof-conj-deadlock}
\end{proof}

\ak{proof trial of the disjoint conj is commented out}

\begin{restatable}[1-Conj, Deadlocks, Unfair]{tightness}{TightConjDeadlockRestricted}
\label{obs:conj:tight_deadlock}
The cutoff $c=2|B|-2$ is tight for parameterized deadlock detection in the 1-conjunctive systems, i.e., for any $k$ there is a system type $(A,B)$ with $|B|=k$ such that there is a deadlock in $(A,B)^{(1,2|B|-2)}$, but not in $(A,B)^{(1,2|B|-3)}$. 
\end{restatable}
\input{tight-conj-deadlock}

\ifwithextensions   
Now consider the general case of guards of the original (unrestricted) form:
\begin{restatable}[Local Deadlock Detection, Unrestricted Guards]{lem}{ConjDeadlockUnrestricted}
\label{le:ConjunctiveBoundingLemmaDeadlocksUnrestricted}
Let $c=3^{\card{B}}$.
If there are no local deadlocks 
in the conjunctive system $\cutoffsys$, then for any $n \ge c$ the conjunctive system $\largesys$ has no local deadlocks.
\end{restatable}

\iffinal
\else
\input{proof-conj-deadlock-unrestricted}
\fi
\fi   

\subsection{Conjunctive Systems with Fairness}
In this section, subscript $i$ in path quantifiers, $\pexists_i$ and $\pforall_i$, 
denotes the quantification over initializing runs.

\begin{restatable}[Monotonicity: Conj, Props, Fair]{lem}{ConjMonLemmaFair}
\label{le:ConjMonFair}
For unconditionally-fair initializing runs of conjunctive systems:\sj{generalization is obvious; $n \ge k+1$ in general case}
\begin{align*}
& \forall n \ge 2:\\
& (A,B)^{(1,n)} \models \pexists_{uncond,i} h(A,B_1)
\ \Impl \
(A,B)^{(1,n+1)} \models \pexists_{uncond,i} h(A,B_1).
\end{align*}
\end{restatable}
\input{proof-conj-mon-fair}

\begin{restatable}[Bounding: Conj, Props, Fair]{lem}{ConjBoundingLemmaFair}
\label{le:FairConjunctiveBounding Lemma}
For unconditionally-fair initializing runs of conjunctive systems:
\begin{align*}
&\forall n \geq 1:\\
& (A,B)^{(1,1)} \models \pexists_{uncond} h(A,B_1)
\ \Implied \
(A,B)^{(1,n)} \models \pexists_{uncond} h(A,B_1)
\end{align*}
\end{restatable}

\input{proof-conj-bounding-fair}

\begin{restatable}[1-Conj, Props, Fair]{tightness}{TightConjPropRestrictedFair}
\label{obs:conj:tight_prop_fair}
The cutoff $c=2$ is tight for parameterized model checking of $\pexists h(A,B_1)$ 
on unconditionally-fair initializing runs in 1-conjunctive systems, 
i.e., 
there is a system type $(A,B)$ and property $\pexists h(A,B_1)$ 
which is satisfied by $(A,B)^{(1,1)}$ but not by $(A,B)^{(1,2)}$.
\end{restatable}
\begin{proof}
The figure below shows $(A,B)$, $\pexists h(A,B_1) = \pexists \FG (b_{init} \impl a_1)$.
\begin{figure}[h]
\centering
\vspace{-20pt}
\subfloat[Template A]{
\centering
\scalebox{0.75}{\input{img/conj_tight_prop_fair_tmplA}}
\label{fig:conj:tight_prop_fair_tmplA}
}
\hspace{1cm}
\subfloat[Template B]{
\centering
\scalebox{0.75}{\input{img/conj_tight_prop_fair_tmplB}}
\label{fig:conj:tight_prop_fair_tmplB}
}
\label{fig:conj:tight_prop_fair_tmpl}
\end{figure}
\end{proof}

\begin{restatable}[Monotonicity: Conj, Deadlocks, Fair]{lem}{ConjDeadlockMonotonicityLemmaFair}
\label{le:FairConjunctiveMonotonicityLemmaDeadlocks}
For 1-conjunctive systems on strong fair initializing or finite runs:
$$
\forall n\geq 1: (A,B)^{(1,n)} \textit{ has a deadlock}
\ \Impl\ 
(A,B)^{(1,n+1)} \textit{ has a deadlock}
$$
\end{restatable}
\begin{proof}\ak{check the minimal value of $n$ (1 or 2?)}
Let $x$ be a globally deadlocked or locally deadlocked strong-fair initializing run of $(A,B)^{(1,n)}$.
We will build a globally deadlocked or locally deadlocked strong-fair initializing run 
of $(A,B)^{(1,n+1)}$.

If $x$ is finite, then $y$ is the copy of $x$, and the new process stays in $\init_B$
until every process become deadlocked, and then is scheduled arbitrarily.
Note that $y$ constructed this way may be locally deadlocked 
rather than globally deadlocked as $x$ is.

Now consider the case when $x$ is locally deadlocked strong-fair initializing.

Let $\mD$ be the set of deadlocked $B$-processes in $x$, and $d$ be the moment 
when the processes become deadlocked.

Consider the case $\visInf{\mB\smi\mD}{x} \neq \emptyset$:
copy $x$ into $y$, and let the new process $B_{n+1}$ wait in $\init_B$ 
and interleave the roles with a process $B$ that moves infinitely often in $x$, 
similarly to as described in the proof of Lemma~\ref{le:ConjMonFair}.

Consider the case $\visInf{\mB\smi\mD}{x} = \emptyset$:
every $B$ process of $(A,B)^{(1,n)}$ is deadlocked and thus $\mD = \mB$.
Define 
$$
DeadGuards = \{\ q \| \exists P \in \mD
                      \textit{ with a transition guarded ``\,}
                      {\forall \neg q} 
                      \textit{\!'' in } (s_d(P),e_d(P))\ \}.
$$
Note that $Q_A \cap DeadGuards = \emptyset$, because $A$ visits infinitely often $\init_A$
and we consider 1-conjunctive systems.
Hence, copy $x$ into $y$, and let the new process $B_{n+1}$ wait in $\init_B$ 
until every process $B_1,...,B_n$ become deadlocked, and then schedule $B_{n+1}$ arbitrarily.
%
\end{proof}

\begin{restatable}[Bounding: 1-Conj, Deadlocks, Fair]{lem}{ConjDeadlockLemmaFair}
\label{le:FairConjunctiveBoundingLemmaDeadlocks}
For 1-conjunctive systems on strong-fair initializing or finite runs:
\ak{no real need for initializing -- but easier to explain}
\li
  \- with $c=2|Q_B\smi \{ \init \}|$ and any $n>c$:
  $$
  \cutoffsys \textit{ has a global deadlock} 
  \ \Implied\ 
  \largesys \textit{ has a global deadlock}
  $$

  \- with $c=2|Q_B\smi \{ \init \}|+1$ and any $n>c$ (when $|Q_B|>2$):
  $$
  \cutoffsys \textit{ has a local deadlock} 
  \ \Implied\ 
  \largesys \textit{ has a local deadlock}
  $$

  \- with $c=2|Q_B\smi \{ \init \}|$ and any $n>c$:
  $$
  \cutoffsys \textit{ has a deadlock} 
  \ \Implied\ 
  \largesys \textit{ has a deadlock}
  $$
\il
\end{restatable}
\begin{proof}
\input{proof-conj-deadlock-fair}
\end{proof}

\begin{restatable}[1-Conj, Deadlocks, Fair]{tightness}{TightConjDeadlockRestrictedFair}
\label{obs:conj:tight_deadlock_fair}
The cutoff $c=2|B|-2$ is tight for deadlock detection on strong-fair initializing
or finite runs in the 1-conjunctive systems, 
i.e., 
for any $k>2$ there is a system type $(A,B)$ with $|B|=k$ such that 
there is a strong-fair initializing deadlocked run in $(A,B)^{(1,2|B|-2)}$, 
but not in $(A,B)^{(1,2|B|-3)}$.
\end{restatable}

\input{tight-conj-deadlock-fair}

\ifwithextensions      
\subsection{Label-dependent Cutoffs for Conjunctive Systems}

To obtain label-dependent cutoffs, we need to re-define the special role of the initial state for conjunctive guards: in this case, we assume that the initial state has a labeling $l_\initstate$, and this labeling must be part of any guard.

\begin{restatable}[Deadlock Detection, Label-based]{lem}{ConjDeadlockLemmaLabel}
\label{le:ConjunctiveBoundingLemmaDeadlocksLabel}
Let $c=2\card{\labelings_B}-2$.
If there are no local deadlocks 
in the conjunctive system $\cutoffsys$, then for any $n \ge c$ the conjunctive system $\largesys$ has no local deadlocks.
\end{restatable}

\iffinal
\else
\input{proof-conj-deadlock-label}
\fi
\fi     

\iffinal
\else
\input{proof-conj-disjoint}
\fi

\iffinal
\else
\input{proof-conj-subsume}
\fi

\iffinal
\else
\section{2-Conjunctive Proof Idea}

\paragraph{Given:} a (initializing) system with only 1-conjunctive guards, except a single 2-conjunctive guard.

\paragraph{Goal:} find cutoff for local deadlock detection.

\paragraph{Idea}:
Assume in a system \largesys there is a locally deadlocked run $x$. That is, at some time $d$ there is a set of processes with state/input combinations $\dead \subseteq Q \times \Sigma$ such that $\forall (q,e) \in \dead$ and $\forall \time \geq d$, no transition from $(q,e)$ is enabled in $x_\time$. Let $\visited_\bot$ be the set of state components of $\dead$, i.e., $\visited_\bot = \{ q \in Q \mid \exists e \in \Sigma: (q,e) \in \dead \}$. 

As in the 1-conjunctive case, this local deadlock may depend on
\begin{enumerate}[label=\alph*)]
\item at least two processes staying in each of the states in $\visited^2_\bot \subset \visited_\bot$, and
\item an additional process staying in each of the states of an additional set $BlockStates$.
\end{enumerate}

Since we want to simulate a strong-fair run $x$, we need to ensure that all processes that are not in $\dead$ at moment $d$ will move infinitely often in our simulation. In particular, we assume that the deadlock depends on the 2-conjunctive guard of the system, and none of the deadlocked processes is in either of the states $q_1$ or $q_2$ that are excluded by this guard (otherwise, the 2-conjunctive guard is not important for the local deadlock and we can fall back to the 1-conjunctive proof).

To allow infinite movement, note that for every $q \in Blockstates$, there must be a loop from $q$ to $\initstate$ and back to $q$ (written $\myloop{q}$) in the original run $x$. While moving along this loop is easy in the pure 1-conjunctive case if we remove all processes not in $\visited_\bot \cup BlockStates$, it is not so easy if we have a 2-conjunctive guard: transitions on the loop $\myloop{q}$ may depend on guards that exclude either $q_1$ or $q_2$, the two states of the 2-conjunctive guard.

First, consider the problem of a process $p$ that occupies (wlog) the state $q_1$ of the 2-conjunctive guard at moment $d$, and has to make a loop $\myloop{q_1}$. The transition out of the loop may depend on a guard that forbids $q_1$, so another process needs to move into $q_2$ before $p$ can move out of $q_1$. The next transition along the loop may depend on a guard that forbids $q_2$, so another process needs to move into $q_1$, and the first additional process then needs to move out of $q_2$, etc.

\paragraph{Question:} How many processes are necessary to simulate just the loop $\myloop{q_1}$?
\fi

%% file: proof-conj-bounding.tex
The proof is inspired by the first part of the proof of \cite[Lemma 5.2]{Emerson00}.

Let $x=(s_1,e_1,p_1) (s_2,e_2,p_2) \ldots$ be a run of $\largesys$. 
Note that by the semantics of conjunctive guards, 
the transitions along any local run of $x$ will also be enabled 
in any system $\cutoffsys$ with $c \leq n$, 
where the processes exhibit a subset of the local runs of $x$. 
Thus, we obtain a run of $\cutoffsys$ by copying a subset of the local runs of $x$, 
and removing elements of the new global run where all processes stutter.
\sj{should we put this as a general lemma somewhere?}

Then, based on an infinite run $x$ of the original system, 
we construct an infinite run $y$ of the cutoff system. 
Let $y(A)=x(A)$ and $y(B_1)=x(B_1)$. 
The second copy of template $B$ in $(A,B)^{(1,2)}$ is needed to ensure that 
the run $y$ is infinite, i.e., at least one process moves infinitely often. 
If both $x(A)$ and $x(B_1)$ eventually deadlock, 
then there exists a process $B_i$ of $\largesys$ that makes infinitely many moves, 
and we set $y(B_2) = x(B_i)$. 
Otherwise, we set $y(B_2) = x(B_2)$.

%% file: img/conj_tight_prop_tmplA.tex
\begin{tikzpicture}[node distance=1.5cm,inner sep=1pt,minimum size=0.5mm,->,>=latex]

\node[state,initial left] (init) {};
\path (init) [loop right] edge [right] node {$\forall \neg 1_B$} (init);

\end{tikzpicture}
  

%% file: img/conj_tight_prop_tmplB.tex
\begin{tikzpicture}[node distance=1.5cm,inner sep=1pt,minimum size=0.5mm,->,>=latex]

\node[state,initial left] (init) {};
\node[state] (b_1) [right= of init] {$1_B$};

\path (init) edge [above] node {} (b_1);
\path (init) [loop below] edge [right] node {} (init);

\end{tikzpicture}
  

%% file: proof-conj-bounding-general.tex
\begin{proof}
The proof is the same as for Lemma~\ref{le:ConjunctiveBoundingLemma}, except that we copy the $k$ local runs that exhibit the property, in addition to the local run that ensures that the resulting global run will be infinite.
\end{proof}

%% file: proof-monotonicity-conj-deadlock.tex
Given a deadlocked run $x$ of $(A,B)^{(1,n)}$, we construct a deadlocked run of $(A,B)^{(1,n+1)}$. Let $y$ copy run $x$, and keep the new process in $\init$. If $x$ is globally deadlocked and $d$ is the moment when the deadlock happens in $x$, then schedule the new process arbitrarily after moment $d$.
%
%
%
%
%

%% file: proof-conj-deadlock.tex
The proof is inspired by the second part of the proof of \cite[Lemma 5.2]{Emerson00}, 
but in addition to global we consider local deadlocks. 

\myparagraph{Global Deadlocks} 
Let $c=2|Q_B\smi\{ \init \}|)$. 
Let run $x = (s_1,e_1,p_1)\ldots(s_d,e_d,\bot)$ of \largesys 
with $n>c$ be globally deadlocked. 
We construct a globally deadlocked run $y$ in $\cutoffsys$:
\li
  \-[a.] for every $q \in Set(s_d) \setminus \{\init\}$:
  \li
    \- if $s_d$ has two processes in state $q$, 
       then devote two processes of \cutoffsys that mimic the behaviour 
       of the two of \largesys correspondingly

    \- otherwise, $s_d$ has only one process in state $q$, 
       then devote one process of \cutoffsys that mimics the process of \largesys
  \il
  \-[b.] for any process of \cutoffsys not used in the construction (if any): 
         let it mimic an arbitrary $B$-process of \largesys 
         not used in the construction (including (b))
\il
The construction uses (if ignore (b)) $\leq 2|Q_B\smi \{ \init \}|$ processes $B$. 
Note that the proof does not assume that the system is 1-conjunctive.

\myparagraph{Local Deadlocks} 
Let $c = |Q_B\smi \{ \init \}|+2$. 
Let run $x = (s_1,e_1,p_1)\ldots$ of \largesys with $n>c$ be locally deadlocked. 
We will construct a run $y$ of \cutoffsys 
where at least one process deadlocks and exactly one process moves infinitely often.

Wlog. we distinguish three cases:
\li
\-[1.] $A$ moves infinitely often in $x$, and $B_1$ deadlocks
\-[2.] $A$ deadlocks, and $B_1$ moves infinitely often
\-[3.] $A$ neither deadlocks nor moves infinitely often, $B_1$ deadlocks, 
       $B_2$ moves infinitely often
\il

\myparagraph{1} ``$A$ moves infinitely often in $x$, and $B_1$ deadlocks''.

Let $q_\bot, e_\bot$ be the deadlocked state and input of $B_1$ in $x$, 
and let $d$ be the moment from which $B_1$ is deadlocked.

Let $DeadGuards=\{q_1,\ldots,q_k\}$ be the set of states
such that for every $q_i \in DeadGuards$ there is an outgoing transitions 
from $q_\bot$ with $e_\bot$ guarded ``$\forall \neg q_i$'',
and assume $DeadGuards \neq \emptyset$
(if it is empty, then we keep every process in $\init$ 
 until someone reaches $q_\bot$ and then schedule the rest arbitrarily). 
(Recall that $q_i \in Q_B \cupdot Q_A$).

The construction is:
\li
  \-[a.] $y(A)=x(A)$, $y(B_1)=x(B_1)$
  \-[b.] for each $q \in DeadGuards$, at moment $d$ in $x$
         there is a process $p_q$ in state $q$. 
         If $p_q \in \{B_1,...,B_n\}$, 
         then let one process of \cutoffsys mimic it till moment $d$, 
         and then stutter in $q$.
  \-[c.] let other processes of \cutoffsys (if any) stay in $\init$.
\il
The construction uses (if ignore (c)) $\leq |Q_B\smi \{ \init \}|+1$ processes $B$.

Note: 
the assumption of 1-conjunctive systems implies that,
in order to deadlock $B_1$,
we need a process in each state in $BlockGuards$.
This implies that having a process in each state of $BlockGuards$ does not disable 
any $A$'s transition after moment $d$.

\myparagraph{2} ``$A$ deadlocks, and $B_1$ moves infinitely often'': 
use the construction from (1).

\myparagraph{3} 
``$A$ neither deadlocks nor moves infinitely often, 
  $B_1$ deadlocks, $B_2$ moves infinitely often''. 
Use the construction from (1), and additionally: $y(B_2)=x(B_2)$. 
Thus, the construction uses (if ignore (c)) $\leq |Q_B \smi \{ \init \}|+2$ 
processes $B$.

\myparagraph{Deadlocks}
Take the higher value among the cases considered above $c=2|Q_B\smi \{ \init \}|$: 
if $x$ is locally deadlocked then the monotonicity lemma ensures 
that there is a deadlocked run in \cutoffsys.

%% file: tight-conj-deadlock.tex
\begin{proof} 
The figure below provides templates $(A,B)$ that proves the observation. In the figure the edge with $\forall{\neg b_1},\ldots,\forall{\neg b_k}$ denotes edges with guards $\forall{\neg b_1},\ldots,\forall{\neg b_k}$. To get the global deadlock we need at least two processes in each $b_i \in \{b_1,\ldots,b_k\}$. Note that the system does not have local deadlocks.\ak{show that cutoffs for local deadlocks are also tight}
\begin{figure}[h]
\vspace{-10pt}
\centering
\subfloat[Template A]{
\centering
\scalebox{0.75}{\input{img/conj_tight_deadlock_tmplA}}
\label{fig:conj:tight_deadlock_tmplA}
}
\hspace{1cm}
\subfloat[Template B]{
\centering
\scalebox{0.75}{\input{img/conj_tight_deadlock_tmplB}}
\label{fig:conj:tight_deadlock_tmplB}
}
\label{fig:conj:tight_dead_tmpl}
\end{figure}
\end{proof}

%% file: img/conj_tight_deadlock_tmplA.tex
\begin{tikzpicture}[node distance=1.8cm,inner sep=1pt,minimum size=0.5mm,->,>=latex]

\node[initial left, state] (init) {};

\path (init) [loop right] edge [right] node {$\forall\neg 1_B$} (init);
\path (init) [loop right,dotted,distance=26mm] edge [right] node {...} (init);
\path (init) [loop right,distance=38mm] edge [right] node {$\forall\neg k_B$} (init);

\end{tikzpicture}
  

%% file: img/conj_tight_deadlock_tmplB.tex
\begin{tikzpicture}[node distance=1.8cm,inner sep=1pt,minimum size=0.5mm,->,>=latex]

\node[initial above, state] (init) {$init$};
\node[state] (b_1) [left=2.2cm of init] {$1_B$};
\node (dots) [below=0.8cm of init] {$\ldots$};
\node[state] (b_k) [right=2.2cm of init] {$k_B$}; 

\path (b_1.20) edge [above] node {$\forall{\neg 1_B},...,\forall{\neg k_B}$} (init.160);
\path (init.200) edge [above] node {} (b_1.340);

\path (b_k.160) edge [above] node {$\forall{\neg 1_B},...,\forall{\neg k_B}$} (init.20); 
\path (init.340) edge [above] node {} (b_k.200); 

\path (init.282) edge [left, dotted] node {} ($(dots)+(0.1,0.1)$);
\path ($(dots)-(0.1,-0.1)$) edge [left, dotted] node {$\forall{\neg 1_B},...,\forall{\neg k_B}$} (init.257);



\end{tikzpicture}
  

%% file: proof-conj-deadlock-unrestricted.tex
\begin{proof}
Suppose run $x= (s_0,e_0,p_0),(s_1,e_1,p_1),\ldots$ of $\largesys$ is locally deadlocked for at least one process, and at least one process keeps on moving forever. 
We construct a run $y= (s^*_0,e^*_0,p^*_0),(s^*_1,e^*_1,p^*_1),\ldots$ where at least one process deadlocks and at least one process moves on forever in the following way.

Let $y(A^1,B^1)=x(A^1,B^j)$, where $j$ is chosen such that either $A^1$ moves infinitely often and $B^j$ deadlocks, or the other way around. Let $U^1$ denote the deadlocked process in $\cutoffsys$.

Then, let $d$ be the point in time from which $A^1$ or $B^i$ are disabled 
forever in $x$. To ensure that all transitions of $U^1$ are disabled in $s_d$ 
of $\cutoffsys$, let $I = \{ i | s_d(B^i) = q \textit{ for some } q \in 
Set(s_d,B)\}$, let $pr: I \rightarrow [2:\card{Set(s_d,B)}]$, and define $y(B^{pr(i)
}) = x(B^i)$ for all $i \in I$. In the following, we want to ensure that: a) 
$U^1$ will always remain deadlocked, and b) all transitions on local run 
$y(\overline{U}^1)$ are enabled. To ensure a), it is sufficient to ensure that 
$Set(s^*_\time,B) \supseteq Set(s_\time,B)$, and to ensure b), it is sufficient to ensure $Set(s^*_\time,B) \subseteq Set(s_\time,B)$. 
Thus, the goal of our 
construction is to ensure $Set(s^*_\time,B) = Set(s_\time,B)$ for all $\time 
\ge d$.\ak{What is $Set$? Is a set of states or set of sets?}

For a transition of $\largesys$ where $A^1$ or $B^j$ move, we fire the same transition in $\cutoffsys$. Otherwise, consider the following cases for a transition of $\largesys$ from $s_\time$ to $s_{\time+1}$:\sj{formally, this is probably a (0,1,many)-counter abstraction}
\begin{enumerate}
\item $Set(s_{\time+1},B) = Set(s_\time,B)$: we drop this transition in $\cutoffsys$, unless it leads to the situation that only one process is in the state that was left in the original transition\sj{i.e., unless it changes the state in the abstraction}; in that case, let all but one processes move along that transition
\item $Set(s_{\time+1},B) = (Set(s_\time,B) \setminus \{ q \}) \cup \{q'\}$ for some $q \in Set(s_\time,B)$, $q' \not\in Set(s_\time,B)$: we simulate the transition from $q$ to $q'$ with the unique process in $\cutoffsys$ that is in state $q$
\item $Set(s_{\time+1},B) = Set(s_\time,B) \setminus \{ q \}$ for some $q \in Set(s_\time,B)$: we simulate the transition from $q$ to a state $q' \in Set(s_{\time+1},B)$ with the unique process in $\cutoffsys$ that is in state $q$
\item $Set(s_{\time+1},B) = (Set(s_\time,B) \cup \{q\}$ for some $q \not\in Set(s_\time,B)$: if possible, we simulate this by a number of transitions from a state $q'$ to $q$, where multiple processes are in $q'$; if this is not possible, we add another local run: $y(B^{c+1})=x(B^i)$ for some $i$ with $s_{\time+1}(B^i) = q$.
\end{enumerate}

Using this construction, we simulate $\largesys$ in $\cutoffsys$ until we have 
reached $\time$ such that there exists $\time'$ with $d < \time' < \time$ and 
$s_\time'=s_\time$. The maximal length of a subsequence of $x$ such that 
$s_\time \neq s_{\time'}$ for all $x_\time,x_\time'$ is bounded by $\card{B}^n$. 
However, we only simulate steps where $Set(s_\time,B)$ changes, or where all 
but one processes move out of a state. Furthermore, we can remove loops 
within the abstraction, i.e., if there are $\time,\time'$ such that 
$Set(s^*_\time)=Set(s^*_{\time'})$, then we cut $x^*[\time+1:\time']$ from the run. 
Thus, there are at most $3^{\card{B}}$ different abstract configurations that 
can be reached on such a path. The number of steps of case d) that add 
another run is thus bounded by $3^{\card{B}}$.\sj{somewhat better cutoff possible, but probably not worth the effort...}

Note that the abstract state at the end of this loop is the same as in the original configuration, and therefore arbitrarily many executions of the loop can be appended.
\end{proof}

%% file: proof-conj-mon-fair.tex
\begin{proof}
Given a unconditionally-fair initializing run $x$ of $\largesys$, we construct a unconditionally-fair initializing run $y$ in $(A,B)^{(1,n+1)}$, with one additional process $p$. 
First, copy all local runs of all processes of $(A,B)^{(1,n)}$ from the run $x$ into $y$.
Then, let process $p'$ stutter in $\init$ until some other process $p \neq B_1$ enters $\initstate$. 
Then, exchange the roles of processes $p'$ and $p$: let $p$ stutter in $\initstate$, while $p'$ takes the transitions of $p$ from the original run, until it enters $\initstate$. And so on.
In this way, we continue to interleave the run between $p'$ and $p$, and obtain a unconditionally-fair initializing run for all processes, with $y(A,B_1)=x(A,B_1)$. 
Thus, if $\largesys \models \pexists h(A,B_1)$, then $(A,B)^{(1,n+1)} \models \pexists h(A,B_1)$.
\end{proof}

%% file: proof-conj-bounding-fair.tex
\begin{proof}
Given an unconditionally-fair [initializing] run $x$ of $\largesys$ with $n>c$ construct an unconditionally-fair [initializing] run $y$ in the cutoff system $(A,B)^{(1,1)}$: copy the local runs of processes $A$, $B_1$.
\end{proof}

%% file: img/conj_tight_prop_fair_tmplA.tex
\begin{tikzpicture}[node distance=1.5cm,inner sep=1pt,minimum size=0.5mm,->,>=latex]

\node[state,initial left] (init) {${init}_A$};
\node[state] (a_1) [right= of init] {$1_A$};

\path (init) edge [above] node {} (a_1);
\path (a_1)  [bend left=20] edge [below] node {} (init);

\end{tikzpicture}
  

%% file: img/conj_tight_prop_fair_tmplB.tex
\begin{tikzpicture}[node distance=1.5cm,inner sep=1pt,minimum size=0.5mm,->,>=latex]

\node[state,initial left] (init) {${init}_B$};
\node[state] (b_1) [right= of init] {$1_B$};
\node[state] (b_2) [right= of b_1] {$2_B$}; 

\path (init) edge [above] node {$\forall\neg 1_B$} (b_1);
\path (b_1)  edge [above] node {$\forall\neg 1_A$} (b_2);
\path (b_2)  [bend left=20] edge [below] node {$\forall\neg 2_B$} (init);

\end{tikzpicture}
  

%% file: proof-conj-deadlock-fair.tex
\providecommand{\deadOne}{\dead_1}
\providecommand{\deadTwo}{\dead_2}

\myparagraph{Global Deadlocks}
$c=2|Q_B \smi \{\init_B\}|$, 
see Lemma~\ref{le:ConjunctiveBoundingLemmaDeadlocks}, 
the fairness does not matter on finite runs.

\myparagraph{Local Deadlocks}
Let $c=2|Q_B\smi \{ \init_B \}|$. 
Let $x= (s_1,e_1,p_1)\ldots$ be a locally deadlocked strong-fair intitializing run 
of $\largesys$ with $n>c$. 
We construct a locally deadlocked strong-fair initializing run $y$ of $\cutoffsys$.

Let $\mD$ be the set of deadlocked processes in $x$. 
Let $d$ be the moment in $x$ starting from which every process in $\mD$ is deadlocked.

Let $\dead(x)$ be the set of states in which processes $\mD$ of \largesys
are deadlocked.

Let $\deadTwo(x) \subseteq \dead(x)$ be the set of deadlocked states such that: 
for every $q \in \deadTwo(x)$, 
there is a process $P \in \mD$ with $s_d(P) = q$ 
and that for input $e_{\geq d}(P)$ has a transition guarded with ``$\forall \neg q$''.
Thus, a process in $q$ is deadlocked with $e_d(P)$
only if there is another process in $q$ in every moment $\geq d$.

Let $\deadOne(x) = \dead(x)\smi\deadTwo(x)$.
I.e., 
for any $q \in \deadOne(x)$, there is a process $P$ of \largesys 
which is deadlocked in $s_d(P) = q$ with input $e_d(P)$,
and no transitions from $q$ with input $e_d(P)$ are guarded with ``$\forall \neg q$''.

Define
$$
DeadGuards = \{\ q \| \exists P \in \mD
                      \textit{ with a transition guarded ``\,}
                      {\forall \neg q} 
                      \textit{\!'' in } (s_d(P),e_d(P))\ \}.
$$
We illustrate properties of sets 
$DeadGuards$, $\deadOne$, $\deadTwo$, $\visInf{\mB\smi\mD}{x}$ 
in Fig.~\ref{fig:conj-deadlocks-venn}.
\ak{check how $A$'s states affect all those sets, currently i assumed that they are all subsets of $Q_B$}

\begin{figure}[h]
\begin{mdframed}
\centering
\includegraphics[width=0.7\textwidth]{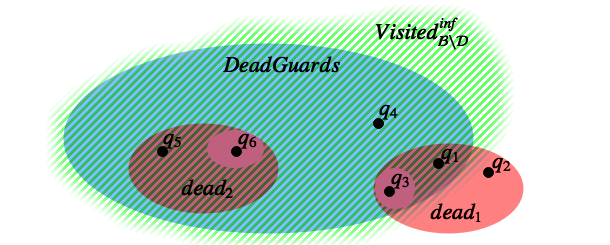}
\captionsetup{singlelinecheck=off}
\caption[fig:conj-deadlocks-venn]{%
Venn diagram for sets $DeadGuards$, $\deadOne$, $\deadTwo$, $\visInf{\mB\smi\mD}{x}$:
\begin{itemize}
\item[($q_1$)] $\deadOne \cap DeadGuards \cap \visInf{\mB\smi\mD}{x} \neq \emptyset$ is possible:
               in $x$, 
               there is a process deadlocked in state $q_1$,
               there is a non-deadlocked process that visits $q_1$ infinitely often,
               and there is a process deadlocked in a state $q \neq q_1$ 
               with a transition guarded ``$\forall \neg q_1$'' 

\item[($q_3$)] $\deadOne \cap DeadGuards \smi \visInf{\mB\smi\mD}{x} \neq \emptyset$ is possible:
               similarly to $q_1$, 
               except that no non-deadlocked processes visit $q_3$ infinitely often

\item[($q_2$)] $\deadOne \smi (\visInf{\mB\smi\mD}{x} \cup DeadGuards) \neq \emptyset$ is possible:
               in $x$, 
               there is a process deadlocked in state $q_2$,
               no other processes visit $q_2$ infinitely often,
               and no processes are deadlocked with a transition guarded ``$\forall \neg q_2$''

\item[($q_4$)] $DeadGuards \smi \dead \neq \emptyset$ is possible:
               there is a process deadlocked in a state $q \neq q_4$ 
               with a transition guarded  ``$\forall \neg q_4$''

\item[($q_5$)] $\deadTwo \cap \visInf{\mB\smi\mD}{x} \cap DeadGuards \neq \emptyset$ is possible:
               there is at least one process deadlocked in $q_5$ with a transition guarded ``$\forall \neg q_5$'',
               and some non-deadlocked process visits $q_5$ infinitely often
               (this process does not deadlock in $q_5$, 
                because in $q_5$ it receives an input different from that of the deadlocked processes)

\item[($q_6$)] $\deadTwo \cap DeadGuards \smi \visInf{\mB\smi\mD}{x} \neq \emptyset$ is possible:
               similarly to $q_5$, except no non-deadlocked processes visit $q_6$ infinitely often
\end{itemize}
}
\label{fig:conj-deadlocks-venn}
\end{mdframed}
\end{figure}

Let us assume $DeadGuards \neq \emptyset$ -- the other case is straightforward.\ak{check}

The construction has two phases, the setup and the looping.
The setup phase is:
\li
\-[a.] $y(A) = x(A)$

\-[b.] for every $q \in \deadOne$: 
   devote one process of \cutoffsys that copies 
   a process of \largesys deadlocked in $q$

\-[c.] for every $q \in \deadTwo \setminus \visInf{\mB\smi\mD}{x}$: 
   devote two processes of \cutoffsys that copy 
   the behaviour of two processes of \largesys that deadlock in $q$

\-[d.] for every $q \in \deadTwo \cap \visInf{\mB\smi\mD}{x}$:
   in $x$, 
   there is a process, $B_q^\inf \in \mB\smi\mD$, that visits $q$ infinitely often,
   and there is a process, $B_q^\bot \in \deadTwo$, deadlocked in $q$.
   Then:
\li
   \-[1.] devote one process of \cutoffsys that copies the behaviour of $B_q^\bot$
   \-[2.] devote one process of \cutoffsys that copies the behaviour of $B_q^\inf$ 
          until it reaches $q$ at a moment after $d$,
          and then provide the same input as $B_q^\bot$ receives at moment $d$.
          This will deadlock the process.
\il

\-[e.] for every $q \in DeadGuards \setminus \dead$:
       note that $q \in \visInf{\mB\smi\mD}{x}$ and, thus, there is a process, 
       $B_q^\inf \in \mB\smi\mD$, 
       that visits $q$ infinitely often.
       Devote one process of \cutoffsys that copies the behaviour of $B_q^\inf$ 
       until it reaches $q$ at a moment after $d$

\-[f.] if $DeadGuards \setminus \dead \neq \emptyset$ 
       or $A \in \mD$,
       then devote one process that stays in $\init_B$.
       The process will be used in the looping phase to ensure that the run $y$ is infinite,
       and that every process of \cutoffsys used in (e) 
       moves infinitely often (and thus $y$ is strong-fair).

\-[g.] let any other process of \cutoffsys (if any) 
       copy behaviour of a process of \largesys 
       that was not used in the construction so far (including this step)
\il
\ak{go through every item, and prove it is necessary (by giving an example)}
The setup phase ensures: 
in every state $q \in \dead$,
there is at least one process deadlocked in $q$ at moment $d$ in $y$. 
Now we need to ensure that the non-deadlocked processes described 
in steps (e) and (f) move infinitely often.

The looping phase is applied to processes in (e) and (f) only\footnote%
{%
  If there are no such processes, then the setup phase produces the sought run $y$.
}.
Order arbitrarily 
$DeadGuards \smi \dead = (q_1,\ldots,q_k) \subseteq \visInf{\mB\smi\mD}{x}$.
Note that $\init_B \not\in (q_1,...,q_k)$.
Let $\mP$ be the set of processes of \cutoffsys used in steps (e) or (f).
Note that $|\mP| = |(q_1,...,q_k)| + 1$.

The looping phase is: set $i=1$, and repeat infinitely the following.
\li
  \- let $P_\init \in \mP$ be the process that is currently in $\init_B$, 
     and $P_{q_i} \in \mP$ -- in $q_i$
     
  \- let $B_{q_i} \in \visInf{\mB\smi\mD}{x}$ be a process of \largesys 
     that visits $q_i$ and $\init_B$ infinitely often.
     Let $P_\init$ of \cutoffsys copy transitions of $B_{q_i}$
     on some path $\init_B \to \ldots \to g_i$,
     then let $P_{g_i}$ copy transitions of $B_{q_i}$ on some path 
     $g_i \to \ldots \to \init_B$. 
     For copying we consider only the paths of $B_{q_i}$ that happen after moment $d$.

  \- $i=i \oplus 1$
\il

The number of copies of $B$ that the construction uses in the worst case is 
(if ignore (g), assume $Q_B>2$, $DeadGuards \smi \dead = \emptyset$, and $A \in \mD$):
$$
1_{(f)} + 2|\deadTwo|_{(c),(d)} + |\deadOne|_{(b)} 
 \leq 
2|Q_B \smi \{\init_B\}| + 1.
$$

\myparagraph{Deadlocks}
The largest value of $c$ among those for ``Local Deadlocks'' 
and for ``Global Deadlocks'' can be used as the sought value of $c$ 
for the case of general deadlocks.
But it will not be the smallest one.
In the proof of the case ``Local Deadlocks'', in the setup phase, 
item (e) can be modified for the case when $A \in \mD$:
since we do not need to ensure that $y$ is infinite, 
we avoid allocating a process in state $\init_B$.
For a given locally deadlocked strong-fair run, the setup phase may produce
the globally deadlocked run, but that is allright for the case of general deadlocks.
With this note, for the general case $c = 2|Q_B \smi \{\init_B\}|$.

%% file: tight-conj-deadlock-fair.tex
\begin{proof} 
Consider the same templates as in Observation~\ref{obs:conj:tight_deadlock}.
%
\end{proof}

%% file: proof-conj-deadlock-label.tex
\begin{proof}
The proof works in the same way as for Lemma~\ref{le:ConjunctiveBoundingLemmaDeadlocks}, except that we consider visited labelings instead of states.
\end{proof}

%% file: proof-conj-disjoint.tex
\subsection{Cutoffs for 2-Conjunctive-Disjoint Guard Systems}
\ak{finished, but not polished}

\begin{restatable}[Deadlock Detection, 2-Disjoint Guards]{lem}{blabla}
For 2-conjunctive-disjoint systems on strong-fair initializing runs, 
with $c\approx 2|Q_B \smi \{\init\}|$\ak{calc} and any $n>c$:
  $$
  \cutoffsys \textit{ has a deadlock} 
  \ \Implied\ 
  \largesys \textit{ has a deadlock}
  $$
\end{restatable}
\begin{proof}
$BlockGuards$ are the guards that are not blocked by the deadlocked processes,
and, thus, must be blocked for the local deadlocks to occur.

Inf-moving processes visit such guards by looping 
$\init \rightsquigarrow \init$.
Thus, through every guard some process goes inf-often.
Link with a guard a process of the large system, called $B_m$, 
that visits it inf-often.

Roughly, the idea of the construction is the following.
We devote one process that starts in $\init$.
For every blocking guard, in the cutoff system, 
a process called $B_m'$ will emulate a loop 
$\init \rightsquigarrow \init$ of process $B_m$ of the large system.
In the loop of the process $B_m$, 
blocking guard states are visited in some order: 
some states are entered through and some are exited through.
First, we ensure that there is a process in every ``exit'' state of every 
blocking guard on the loop of $B_m$.
Second, let $B_m'$ mimick transitions of $B_m$ till it reaches the first
``enter'' state, 
then let the process, that is currently in the ``exit'' state, 
continue mimicking the transitions of $B_m$, and so on.
This way we reach the blocking guard under consideration and then reach $\init$.

The realitiy is more involved: 
the main complication is that the same blocking guard may be visited 
several times, thus, strictly speaking, ``enter'' and ``exit'' states
are not well-defined in the above.
Below we decribe the details.

\smallskip
\noindent
\emph{Notes.}
Since there is only a finite number of transitions and blocking guards 
but runs are infinite, after some moment:
\li
\-[1.] Any loop's transition fires inf-often.
\-[2.] Any loop's path between two blocking guards fire inf-often.
\il

Below we consider only such looping paths.

\smallskip
\noindent
\emph{Simplification.}
Given a blocking guard $g$ and a loop that visits $g$, 
divide the loop into three parts:
\li
\- $\init \rightsquigarrow g$ 
   (from $\init$ until the first visit of $g$),
\- $g \rightsquigarrow g$ 
   (from the first visit of $g$ until the last visit), and
\- $g \rightsquigarrow \init$ 
   (from the last visit until $\init$).
\il

Given a guard and a part, \emph{enter[exit]} state is 
the first[last] state of the guard visited on the path.
For the middle part, 
this equivalents to the first[last] state of the path.

Now the \emph{simplification}.
Given a part of the loop, 
remove all loops from the part that start and end in the same state.
Apply this to the first, middle, and last part.

%

Call the loop consisting of the simplified first, middle and last parts,
the \emph{simplified loop}.
Then, the simplification ensures:
\li
\- for every simplified part: 
   a blocking guard $g$ can be visited more than once
   only if $enter_g \neq exit_g$
   and the internal transitions of $g$ are not used.

\- in the simplified loop:
   if an internal blocking guard transition is fired more than once,
   then the firings are separated by a visit of another blocking guard.
\il

\smallskip
\noindent
\emph{Preparation.}
Fix a part of the simplified loop (first, middle, or last).
In this part, for each guard $g$, the transition through $g$ 
is of the form:
\li
\-[f1.] $enter \rightarrow exit$ ($g$ is visited once), or
\-[f2.] $enter \rightarrow \neg g \rightsquigarrow exit$ ($g$ is visited twice).
\il
Given a blocking guard $g$ visited in the part,
the \emph{preparation of $g$ wrt. the part} is:
if transition through $g$ is of the second form -- do nothing,
if of the first form -- wait until the transition is fired
in the original run of the large system, then execute it.
The preparation ensures:
\li
\- if the process $B_g'$ is in state $enter$, 
   then in the part the transition through the guard is of the second form.
\il

\smallskip
\noindent
\emph{Main construction.}
For each blocking guard, fix a loop that visits it inf-often.
Fix a blocking guard $guard$, called the target guard, and, thus, fix the loop.
Derive the simplified loop.
Let the simplified loop be of the form:
$$ 
\underbrace{\init \ g_1' \ ... \ g_f'}_\text{the first part} \ 
\underbrace{g^\star \ g_1'' \ ... \ g_m'' \ g^\star}_\text{the middle part} \ 
\underbrace{g_1''' \ ... \ g_l''' \ \init}_\text{the last part}
$$
In the above formula, $f$, $m$, or $l$ can be $0$, but for simplicity
consider they are not.


Start with $\init$ and the process $B_\init'$.
Call the process of the cutoff system, we currently move, $B_m'$.
Initially $B_m'$ is $B_\init$.
Prepare the process in $g_1'$ wrt. the first part.
Make $B_m'$ mimick $B_m$ until it reaches $g_1'$:
\li
\- if transition through $g_1'$ in the first part is of form (f2),
   then let $B_m'$ continue mimicking $B_m$ until it reaches 
   the next blocking guard
\- if transition through $g_1'$ in the first part if of form (f1),
   then leave $B_m'$ in $enter_{g_1'}$, wait until $B_m$ reaches $exit_{g_1'}$ 
   (this ensures that we do not transit inside the guard 
   when there is another process, $B_{exit_{g_1'}}$, in it).
   Then, prepare the next blocking guard wrt. the corresponding part.
   Set $B_m' = B_{exit_{g_1'}}$.
   Then, let $B_m'$ mimick $B_m$ until it reaches the next blocking guard.
\il
Repeat this construction in a natural way until $B_m'$ reaches $\init$.

For a given target guard $g^\star$, the main construction ensures that the process 
$B_{g^\star}$ moves to the next blocking guard $g$ on the simplified loop.
But we need to ensure that $B_{q^\star}$ eventually reaches $\init$.
To this end, set $g^\star=g$ and repeat the construction.
And so on, until the desired process reaches $\init$.
Finally, do this for every blocking guard.

This concludes the description of the main construction.

Note: the setup phase of the main construction, 
      where we put one process into every blocking guard and 
      one or two processes into every deadlocked state,
      is straightforward.

Finally, the cutoff is defined by the maximal number of the blocking guards + 1,
and is of the scale $\approx 2|Q_B\smi\{\init\}|$.
\end{proof}

%% file: proof-conj-subsume.tex
\subsection{Cutoffs for Conjunctive-Embedding Guard Systems}
\ak{not finished -- run into the problem that I could not solve.
    Originally, I wanted to put every process into every guard, 
    into the least restrictive state, 
    and then do looping, but the problem happens if the transition inside the guard is self-guarded}

\begin{restatable}[Deadlock Detection, Embedding guards]{lem}{blabla}
For 2-conjunctive-embedding systems on strong-fair initializing runs, 
with $c=???$ and any $n>c$:
  $$
  \cutoffsys \textit{ has a deadlock} 
  \ \Implied\ 
  \largesys \textit{ has a deadlock}
  $$
\end{restatable}
\begin{proof}
Given a locally deadlocked strong-fair initializing run of the \largesys, 
we construct such run in \cutoffsys.

Examples of 2-conjunctive-embedding guards: 
$\{\{a,b\}, \{a\}\}$, but not $\{\{a,b\}, \{a\}, \{b\}\}$.
I.e., all guards that are embedded into some guard are comparable with $\subset$
\footnote{Why guards like $(a,b), (a), (b)$ are more difficult and $(a,b),(b)$ are easier?
          Because if $(a,b)$ needs to blocked (but not $(a)$ or $(b)$),
          then a system run may be such that forces processes to flit $a \leftrightarrow b$ 
          when some other process goes from $\init$ to $(a,b)$.
          And if the process template has only $a \rightarrow b$, 
          then the process that goes $\init \rightarrow (a,b)$ can throw the process in $(a,b)$ 
          out of the guard 
          (first, move it $a \rightarrow b$ with transition guarded $\forall \neg a$, 
           second, move it out of $b$ with $\forall \neg b$).
          In contrast, in embedded guard systems when the process moves $\init \rightarrow (a,b)$,
          it is safe to keep the blocking process in $a$ until some moment, 
          when we can move it into $b$ and then finally out of $(a,b)$.}
Also, for now, disallow more than one level of embedding: 
if $g_1 \subset g_2$, then no other $g_i \subset g_1$.

The construction is based on three notes:
\li
\-[1.] invariant state: every cycle starts and ends with the some special state of $B$ processe, 
       called invariant state
\-[2.] abstraction of transitions inside the guard:
       instead of literal repeating how processes move inside the guard,
       we repeat only ``important'' transitions
\-[3.] loops repeat: we might need two loops from the original run of the large system
       in order to complete one loop in the cutoff system
\-[4.] there is a special case of inside-the-guard transitions
\il

We copy local runs of deadlocked processes as usual.\ak{TODO}
For the deadlocked processes to be deadlocked, 
we need to block the guards $BlockGuards' = \{ g_1, ..., g_{k'} \}$.
Some of those guards are already blocked by the deadlocked processes, 
so let $BlockGuards=\{g_1, ..., g_k\} \subseteq BlockGuards'$ denote the guards that are blocked
by infinitely moving processes of the large system.
Thus, our task is to construct a run of \cutoffsys such that:
always there is a process in some state of every $g_1,...,g_k$,
\emph{and} such processes move infinitely often.

\myparagraph{1) Invariant state}

What are the loop transitions? 
Take an inf-moving process in the large system.
It visits \init inf-often, 
and it visits some of guards $\subseteq BlockGuards$ inf-often\footnotemark[10].
Thus, for every guard $g \in BlockGuards$ there is a process of the large system
that moves $\init \rightsquigarrow g \rightsquigarrow \init$ inf-often.

There may be loops of several types. 
Consider the loop type in the figure.
In this loop, 
a process enters guard $\{a,b\}$ via state $a$ and then singleton-guard $\{b\}$.
Call this process $B_m$ (``moving'' process), 
and let $B_m'$ be the process of the cutoff system that currently copies its transitions.
When $B_m$ (and $B_m'$) is outside of $\{a,b\}$ while moving along the loop towards the guard,
in the large system there is some process that blocks the guard.
In the cutoff system, suppose we have a process $B_{a,b}'$ in state $a$ 
at the moment when $B_m$ (and thus $B_m'$) is in $\init$\footnotemark[666].

Process $B_m$ moves along the loop into the guard, 
and some transitions may be guarded $\forall\neg b$\footnotemark[20]
(that is why we cannot place $B_{a,b}'$ into $b$).
What we would like is to move $B_m'$ into $\{a,b\}$ by copying transitions of $B_m$,
then move $B_{a,b}'$ out of $\{a,b\}$ using transitions of $B_m$,
while leaving $B_m'$ in ${a,b}$.

Note 1: if process $B_{a,b}'$ stays in $a$ until process $B_m'$ (and $B_m$) reaches $a$,
then the transitions of $B_m'$ on path $\init \rightsquigarrow a$ are enabled.
However, keeping $B_{a,b}$ in $a$ may result in disabling the subsequent transition
of $B_m$ from $a$ into $b$
(imagine $B_m$ does $\transition{a}{b}{\forall\neg \{a,b\}}$ ---
 $B_{a,b}'$ of the cutoff system cannot repeat this).

Since $B_m$ reaches $a$ (and has a transition guarded $\forall \neg b$), 
there is the last moment in the path $\init \rightsquigarrow a$ 
when $\forall \neg b$ needs to be enabled\footnotemark[200].
After that moment, we can move $B_{a,b}'$ from $a$ into $b$.
But this transition should be enabled.
Now two cases are possible:
\li
\- the transition $a \rightarrow b$ becomes enabled 
   while $B_m$ moves $\init \rightsquigarrow a$.
   Hence, we move process $B_{a,b}'$ when this happens for the last time.
   \footnotemark[333]

\- the transition $a \rightarrow b$ became enabled 
   before $B_m$ started $\init \rightsquigarrow a$.
\il

\hl{the following scenario is problemic: 
    imagine in the large system:
    $B_{a,b}$ stays in $a$, until $B_m$ reaches $a$, 
    then $B_{a,b}$ goes outside for a moment -- $B_m$ goes into $b$ -- $B_{a,b}$ returns into $a$, 
    and finally $B_m$ leaves.
    Later, $B_{a,b}$ will move to make it strong-fair. 
    Here is the picture ($a \rightarrow b$ is guarded $\forall \neg \{a,b\}$):}
    \includegraphics[width=2.5cm]{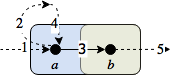}

\footnotetext[10]{If no such processes exists, copy one such process, and we are done.}
\footnotetext[20]{If no such transitions exists, we put $B_{a,b}$ into $b$, 
                  and the construction is similar to 1-conj guards looping, 
                  because the processes can pass the baton of copying $B_m$ when they meet in $b$.
                  Almost: what if there is $\transition{}{}{\forall\neg \{a,b\}}$?\ak{todo}}
\footnotetext[200]{\ak{account for the case when it was enabled before we started from \init}}
\footnotetext[333]{\ak{Here we use the assumption that process $B_{a,b}$ 
                       (of the large system) moved $a \rightarrow b$.}}
\footnotetext[666]{Careful here -- account for different schedulings and possibilities.}

\end{proof}